\documentclass[letterhpaper, 11pt, onecolumn, draftcls]{IEEEtran}

\usepackage{amsfonts}
\usepackage{mathrsfs}
\usepackage{amsfonts}
\usepackage{amssymb}
\usepackage{graphicx}
\usepackage{psfrag}
\usepackage{amsmath}
\usepackage{array}
\usepackage{cases}
\usepackage{color}

\newtheorem{Problem}{Problem}
\newtheorem{Algorithm}{Algorithm}
\newtheorem{Subproblem}{Subproblem}

\begin{document}
\newtheorem{theorem}{Theorem}
\newtheorem{acknowledgement}[theorem]{Acknowledgement}
\newtheorem{axiom}[theorem]{Axiom}
\newtheorem{case}[theorem]{Case}
\newtheorem{claim}[theorem]{Claim}
\newtheorem{conclusion}[theorem]{Conclusion}
\newtheorem{condition}[theorem]{Condition}
\newtheorem{conjecture}[theorem]{Conjecture}
\newtheorem{criterion}[theorem]{Criterion}
\newtheorem{definition}[theorem]{Definition}
\newtheorem{example}[theorem]{Example}
\newtheorem{exercise}[theorem]{Exercise}
\newtheorem{lemma}{Lemma}
\newtheorem{corollary}{Corollary}
\newtheorem{notation}[theorem]{Notation}
\newtheorem{problem}[theorem]{Problem}
\newtheorem{proposition}{Proposition}
\newtheorem{solution}[theorem]{Solution}
\newtheorem{summary}[theorem]{Summary}
\newtheorem{assumption}{Assumption}
\newtheorem{examp}{\bf Example}
\newtheorem{probform}{\bf Problem}
\def\remark{{\noindent \bf Remark:\hspace{0.5em}}}

\def\qed{$\Box$}
\def\QED{\mbox{\phantom{m}}\nolinebreak\hfill$\,\Box$}
\def\proof{\noindent{\emph{Proof:} }}
\def\poof{\noindent{\emph{Sketch of Proof:} }}
\def
\endproof{\hspace*{\fill}~\qed
\par
\endtrivlist\unskip}
\def\endproof{\hspace*{\fill}~\qed\par\endtrivlist\vskip3pt}

\def\E{\mathbb{E}}
\def\eps{\varepsilon}
\def\phi{\varphi}
\def\Lsp{{\boldsymbol L}}
\def\Bsp{{\boldsymbol B}}
\def\lsp{{\boldsymbol\ell}}
\def\Ltsp{{\Lsp^2}}
\def\Lpsp{{\Lsp^p}}
\def\Linsp{{\Lsp^{\infty}}}
\def\LtR{{\Lsp^2(\Rst)}}
\def\ltZ{{\lsp^2(\Zst)}}
\def\ltsp{{\lsp^2}}
\def\ltZt{{\lsp^2(\Zst^{2})}}
\def\ninN{{n{\in}\Nst}}
\def\oh{{\frac{1}{2}}}
\def\grass{{\cal G}}
\def\ord{{\cal O}}
\def\dist{{d_G}}
\def\conj#1{{\overline#1}}
\def\ntoinf{{n \rightarrow \infty }}
\def\toinf{{\rightarrow \infty }}
\def\tozero{{\rightarrow 0 }}
\def\trace{{\operatorname{trace}}}
\def\ord{{\cal O}}
\def\UU{{\cal U}}
\def\rank{{\operatorname{rank}}}
\def\acos{{\operatorname{acos}}}

\def\SINR{\mathsf{SINR}}
\def\SNR{\mathsf{SNR}}
\def\SIR{\mathsf{SIR}}
\def\tSIR{\widetilde{\mathsf{SIR}}}
\def\Ei{\mathsf{Ei}}
\def\l{\left}
\def\r{\right}
\def\({\left(}
\def\){\right)}
\def\lb{\left\{}
\def\rb{\right\}}

\setcounter{page}{1}

\newcommand{\eref}[1]{(\ref{#1})}
\newcommand{\fig}[1]{Fig.\ \ref{#1}}

\def\bydef{:=}
\def\ba{{\mathbf{a}}}
\def\bb{{\mathbf{b}}}
\def\bc{{\mathbf{c}}}
\def\bd{{\mathbf{d}}}
\def\bee{{\mathbf{e}}}
\def\bff{{\mathbf{f}}}
\def\bg{{\mathbf{g}}}
\def\bh{{\mathbf{h}}}
\def\bi{{\mathbf{i}}}
\def\bj{{\mathbf{j}}}
\def\bk{{\mathbf{k}}}
\def\bl{{\mathbf{l}}}
\def\bm{{\mathbf{m}}}
\def\bn{{\mathbf{n}}}
\def\bo{{\mathbf{o}}}
\def\bp{{\mathbf{p}}}
\def\bq{{\mathbf{q}}}
\def\br{{\mathbf{r}}}
\def\bs{{\mathbf{s}}}
\def\bt{{\mathbf{t}}}
\def\bu{{\mathbf{u}}}
\def\bv{{\mathbf{v}}}
\def\bw{{\mathbf{w}}}
\def\bx{{\mathbf{x}}}
\def\by{{\mathbf{y}}}
\def\bz{{\mathbf{z}}}
\def\b0{{\mathbf{0}}}

\def\bA{{\mathbf{A}}}
\def\bB{{\mathbf{B}}}
\def\bC{{\mathbf{C}}}
\def\bD{{\mathbf{D}}}
\def\bE{{\mathbf{E}}}
\def\bF{{\mathbf{F}}}
\def\bG{{\mathbf{G}}}
\def\bH{{\mathbf{H}}}
\def\bI{{\mathbf{I}}}
\def\bJ{{\mathbf{J}}}
\def\bK{{\mathbf{K}}}
\def\bL{{\mathbf{L}}}
\def\bM{{\mathbf{M}}}
\def\bN{{\mathbf{N}}}
\def\bO{{\mathbf{O}}}
\def\bP{{\mathbf{P}}}
\def\bQ{{\mathbf{Q}}}
\def\bR{{\mathbf{R}}}
\def\bS{{\mathbf{S}}}
\def\bT{{\mathbf{T}}}
\def\bU{{\mathbf{U}}}
\def\bV{{\mathbf{V}}}
\def\bW{{\mathbf{W}}}
\def\bX{{\mathbf{X}}}
\def\bY{{\mathbf{Y}}}
\def\bZ{{\mathbf{Z}}}

\def\mA{{\mathbb{A}}}
\def\mB{{\mathbb{B}}}
\def\mC{{\mathbb{C}}}
\def\mD{{\mathbb{D}}}
\def\mE{{\mathbb{E}}}
\def\mF{{\mathbb{F}}}
\def\mG{{\mathbb{G}}}
\def\mH{{\mathbb{H}}}
\def\mI{{\mathbb{I}}}
\def\mJ{{\mathbb{J}}}
\def\mK{{\mathbb{K}}}
\def\mL{{\mathbb{L}}}
\def\mM{{\mathbb{M}}}
\def\mN{{\mathbb{N}}}
\def\mO{{\mathbb{O}}}
\def\mP{{\mathbb{P}}}
\def\mQ{{\mathbb{Q}}}
\def\mR{{\mathbb{R}}}
\def\mS{{\mathbb{S}}}
\def\mT{{\mathbb{T}}}
\def\mU{{\mathbb{U}}}
\def\mV{{\mathbb{V}}}
\def\mW{{\mathbb{W}}}
\def\mX{{\mathbb{X}}}
\def\mY{{\mathbb{Y}}}
\def\mZ{{\mathbb{Z}}}

\def\cA{\mathcal{A}}
\def\cB{\mathcal{B}}
\def\cC{\mathcal{C}}
\def\cD{\mathcal{D}}
\def\cE{\mathcal{E}}
\def\cF{\mathcal{F}}
\def\cG{\mathcal{G}}
\def\cH{\mathcal{H}}
\def\cI{\mathcal{I}}
\def\cJ{\mathcal{J}}
\def\cK{\mathcal{K}}
\def\cL{\mathcal{L}}
\def\cM{\mathcal{M}}
\def\cN{\mathcal{N}}
\def\cO{\mathcal{O}}
\def\cP{\mathcal{P}}
\def\cQ{\mathcal{Q}}
\def\cR{\mathcal{R}}
\def\cS{\mathcal{S}}
\def\cT{\mathcal{T}}
\def\cU{\mathcal{U}}
\def\cV{\mathcal{V}}
\def\cW{\mathcal{W}}
\def\cX{\mathcal{X}}
\def\cY{\mathcal{Y}}
\def\cZ{\mathcal{Z}}
\def\cd{\mathcal{d}}
\def\Mt{M_{t}}
\def\Mr{M_{r}}
\def\O{\Omega_{M_{t}}}
\newcommand{\figref}[1]{{Fig.}~\ref{#1}}
\newcommand{\tabref}[1]{{Table}~\ref{#1}}

\newcommand{\var}{\mathsf{var}}
\newcommand{\fb}{\tx{fb}}
\newcommand{\nf}{\tx{nf}}
\newcommand{\BC}{\tx{(bc)}}
\newcommand{\MAC}{\tx{(mac)}}
\newcommand{\Pout}{P_{\mathsf{out}}}
\newcommand{\nnn}{\nn\\}
\newcommand{\FB}{\tx{FB}}
\newcommand{\TX}{\tx{TX}}
\newcommand{\RX}{\tx{RX}}
\renewcommand{\mod}{\tx{mod}}
\newcommand{\m}[1]{\mathbf{#1}}
\newcommand{\td}[1]{\tilde{#1}}
\newcommand{\sbf}[1]{\scriptsize{\textbf{#1}}}
\newcommand{\stxt}[1]{\scriptsize{\textrm{#1}}}
\newcommand{\suml}[2]{\sum\limits_{#1}^{#2}}
\newcommand{\sumlk}{\sum\limits_{k=0}^{K-1}}
\newcommand{\eqhsp}{\hspace{10 pt}}
\newcommand{\tx}[1]{\texttt{#1}}
\newcommand{\Hz}{\ \tx{Hz}}
\newcommand{\sinc}{\tx{sinc}}
\newcommand{\tr}{\mathrm{tr}}
\newcommand{\diag}{\mathrm{diag}}
\newcommand{\MAI}{\tx{MAI}}
\newcommand{\ISI}{\tx{ISI}}
\newcommand{\IBI}{\tx{IBI}}
\newcommand{\CN}{\tx{CN}}
\newcommand{\CP}{\tx{CP}}
\newcommand{\ZP}{\tx{ZP}}
\newcommand{\ZF}{\tx{ZF}}
\newcommand{\SP}{\tx{SP}}
\newcommand{\MMSE}{\tx{MMSE}}
\newcommand{\MINF}{\tx{MINF}}
\newcommand{\RC}{\tx{MP}}
\newcommand{\MBER}{\tx{MBER}}
\newcommand{\MSNR}{\tx{MSNR}}
\newcommand{\MCAP}{\tx{MCAP}}
\newcommand{\vol}{\tx{vol}}
\newcommand{\ah}{\hat{g}}
\newcommand{\tg}{\tilde{g}}
\newcommand{\teta}{\tilde{\eta}}
\newcommand{\heta}{\hat{\eta}}
\newcommand{\uh}{\m{\hat{s}}}
\newcommand{\eh}{\m{\hat{\eta}}}
\newcommand{\hv}{\m{h}}
\newcommand{\hh}{\m{\hat{h}}}
\newcommand{\Po}{P_{\mathrm{out}}}
\newcommand{\Poh}{\hat{P}_{\mathrm{out}}}
\newcommand{\Ph}{\hat{\gamma}}
\newcommand{\mat}[1]{\begin{matrix}#1\end{matrix}}
\newcommand{\ud}{^{\dagger}}
\newcommand{\C}{\mathcal{C}}
\newcommand{\nn}{\nonumber}
\newcommand{\nInf}{U\rightarrow \infty}

\title{Decentralized Fair Scheduling in Two-Hop Relay-Assisted Cognitive OFDMA Systems}

\author{Rui Wang$^\ast$,  Vincent K. N. Lau, and Ying Cui\\
Department of Electronic \& Computer Engineering\\
Hong Kong University of Science \& Technology\\
 Clear Water Bay, Hong Kong\\
Email:  ray.wang.rui@gmail.com, eeknlau@ust.hk, cuiying@ust.hk }

\maketitle

\begin{abstract}
In this paper, we consider a two-hop relay-assisted cognitive downlink OFDMA system (named as secondary system)
dynamically accessing a spectrum licensed to a primary network, thereby improving the efficiency of spectrum usage. A
cluster-based relay-assisted architecture is proposed for the secondary system, where relay stations are employed for
minimizing the interference to the users in the primary network and achieving fairness for cell-edge users. Based on
this architecture, an asymptotically optimal solution is derived for jointly controlling data rates, transmission
power, and subchannel allocation to optimize the average weighted sum goodput where the proportional fair scheduling (PFS)
is included as a special case. This solution supports decentralized implementation, requires small communication
overhead, and is robust against imperfect channel state information at the transmitter (CSIT) and sensing measurement.
The proposed solution achieves significant throughput gains and better user-fairness compared with the existing
designs. Finally, we derived a simple and asymptotically optimal scheduling solution as well as the associated
closed-form performance under the proportional fair scheduling for a large number of users. The system throughput is
shown to be $\mathcal{O}\left(N(1-q_p)(1-q_p^N)\ln\ln K_c\right)$, where $K_c$ is the number of users in one cluster,
$N$ is the number of subchannels and $q_p$ is the active probability of primary users.

\end{abstract}

{\em {\bf EDICS Items}: WIN-CLRD, WIN-CONT.}

\newpage

\section{Introduction}
Dynamic spectrum access \cite{Akyildiz:NextGenDynamSpectAccessCognitiveRadioSurvey:2006} is a new paradigm to meet the challenge
of the rapidly growing demands of broadband access and the spectrum scarcity for designing the next-generation wireless
communication systems. This motivates the study in this paper on designing a two-hop relay-assisted cognitive OFDMA system which
dynamically shares spectrum access with a primary system (PU) by exploiting its idle periods.

\subsection{Related Work and Motivation}

The issues of power control for dynamic spectrum access in ad hoc networks are addressed in
\cite{Huangjianwei:06,Sharma:07,QianchuanZhao:08}. Cellular systems using cognitive radio for dynamically accessing the
television spectrum are being standardized by the IEEE 802.22 working group. In \cite{Islam:08}, a joint beamforming and power
control algorithm is proposed for a cognitive cellular systems to mitigate interference to the primary network. A key obstacle
for implementing dynamic spectrum access in cellular systems is that direct transmission from base stations to cell-edge users
requires large power and thus causes strong interference to the users in the primary networks. As a result, the users in the
cell-edge will have very small access opportunity due to the primary user activities and this fairness issue cannot be solved by
simply fair scheduling at the base station because the users on the cell edge is limited by the channel access opportunity rather
than the scheduling opportunity. Hence, relay-assisted cellular system will be an effective solution for alleviating the above
fairness issue because it helps to reduce the transmission power required to reach the mobiles on the cell edge. However, there
are still a few critical issues associated with the design and operation of relay-assisted CR systems as summarized below.
\begin{itemize}
\item {\bf{\textcolor{black}{Optimal Decentralized Power, Rate and Subchannel Allocation Algorithm}}}: \textcolor{black}{Extensive research has been carried out on resource allocation in point-to-point relay-assisted communication systems. Power and subchannel allocations for relay-assisted OFDMA systems are studied in
\cite{WooseokNam:07,Oyman:07,GuoqingLi:06,NgChiuYam:07,Awad:08,Can:08,ZhihuaTang:09,Calcev:09,HongxingLi:09,Park:09,RS:10}.
However, these existing works consider centralized solution (e.g. at
BS) in which the resource (power, rate and subchannel) allocation of
the BS and the RSs is computed in a centralized manner at the BS
based on the global system state knowledge
\footnote{\textcolor{black}{Global system state refers to the
aggregate of the channel state information (CSI) of all the BS-RS
links, the RS-MS links, the BS-MS links as well as the sensing
measurements of the BS and all the relays.}}.  Hence, the
conventional centralized approach is very difficult to implement in
practice due to huge signaling overhead and computational
complexity. Moreover, various simplifying assumptions were made in
these literatures to simplify the resource allocation problem in
2-hop OFDMA systems at the cost of performance loss. For example,
one typical constraint is that the relay can only receive the data
for one MS in each subchannel, and this data should be forwarded
completely and exclusively to the target MS in one subchannel in
next phase \cite{Awad:08,ZhihuaTang:09}. This may cause significant
performance loss when the BS-RS link is much better than the RS-MS
link. Therefore, the challenge is to have a decentralized solution
\footnote{\textcolor{black}{By decentralized, we mean the resource
control actions at the BS and the $M$ RSs are computed locally at
the BS and each of the $M$ RS respectively based on the {\em local
system state} at each nodes. There are also explicit message passing
between the BS and the M RS nodes.  Local system state at the BS
refers to the CSI of the BS-mobile, BS-relay links and the sensing
measurement of the BS; local system state at the $ m $-th RS refers
to CSI of the $ m $-th RS to all its MSs and the sensing measurement
of the $ m $-th RS. Thus, the global system state is the aggregation
of local system states at BS and all the $ M $ relays.}} without
performance loss compared with the centralized solutions.}

\item {\bf{Fairness Consideration in Two Hop Systems}}: Conventional relay-assisted
cellular systems perform resource allocation to maximize the sum-throughput \cite{WooseokNam:07,Oyman:07}. Yet, fairness is an
important requirement and a general solution of fair scheduling in relay-assisted (two-hop) CR system is still not fully
addressed. When fairness is considered in a relay-assisted system, neither the optimization objective nor the flow balance
constraint for the relays is convex. Therefore, the conventional approaches for the sum-throughput optimization in the previous
works cannot be applied, and how to solve such resource allocation problem with fairness consideration in relay-assisted systems
is an important challenge to overcome.

\item {\bf{Dynamic Spectrum Access with Imperfect CSIT and Sensing
Measurement}}: In conventional resource optimization problems in
relay-assisted systems \cite{WooseokNam:07,Oyman:07}, there is no
consideration on dynamic spectrum sharing aspects. However, the
presence of PU activity and dynamic spectrum sharing has changed
the fundamental dynamics of the resource allocation problem. For
efficient spectrum sharing, it is critical for the CR systems to
be able to exploit the temporal and spatial burstiness of the PU
activity gaps and yet at the same time, without interrupting the
PU transmissions. This problem is even more challenging when we
have to take into account the imperfect  channel state information
and sensing measurement in which interference to the PU cannot be
completely avoided.

\end{itemize}

\subsection{Contributions}
The key contributions of our work are summarized as follows.
\textcolor{black}{We consider a cluster-based two-hop RS-assisted
cognitive OFDMA system, as shown in Figure \ref{Fig:model} and
\ref{Fig:frame}. We are interested in the associated resource
control problem, which is a difficult  non-convex problem. Moreover,
traditional centralized optimization approach requires significant
communication overhead between the base station and the relay
stations, and has exponentially many control variables w.r.t. the
number of independent subchannels.} In order to tackle these
difficulties, we divide and conquer the resource control problem
into a base station master problem and the relay station
subproblems, where the number of control variables is significantly
reduced (grows linearly w.r.t. the number of frequency bands). We
derive a low-complexity, low-overhead and decentralized algorithm
for controlling power, rate, and subchannel allocation, which
asymptotically maximizes the weighted sum goodput
(\textcolor{black}{average b/s/Hz {\em successfully received} by the
MS}) under the primary-user interference constraint.
\textcolor{black}{We also include the well-known proportional fair
scheduling (PFS) as a special case in our formulation.} The solution
accounts for multiuser diversity, user fairness, imperfect channel
state information at the transmitter (CSIT) and spectrum sensing. As
shown by simulations, the proposed resource allocation algorithm
significantly improves the fairness for cell-edge users. Finally, a
simple and asymptotically optimal scheduling policy as well as the
closed-form performance for PFS is derived to obtain design
insights. For instance, we show that the throughput of the proposed
two-hop relay-assisted cognitive OFDMA system under PFS is
$\mathcal{O}\left(N(1-q_p)(1-q_p^N)\ln\ln K_c\right)$, where $K_c$
is the number of users in one cluster, $N$ is the number of
independent subchannels and $q_p$ is the active probability of
primary users on one subchannel.

The remainder of this paper is organized as follows. The system model is described in Section~\ref{Section:SysModel}. In
Section~\ref{Section:ProbForm},  the problem of optimal power, rate, and subchannel allocation is formulated; the solutions are
presented in Section~\ref{Section:Solution}. Asymptotic throughput analysis is given in Section~\ref{Section:AsympAna}.
Section~\ref{Section:Sim} contains simulation results, followed by concluding remarks in Section~\ref{Section:Conclusion}.

\section{System Model}\label{Section:SysModel}

\subsection{Architecture and Protocol}
As illustrated in Figure \ref{Fig:model}, the secondary user (SU)
system is a cluster-based relay-assisted cognitive OFDMA downlink
system consists of one base station (BS) transmitting  to $K$ mobile
users (MS), where communications are assisted by $M$ relay stations
(RS) as elaborated shortly. The cell is divided into $M+1$ clusters
as shown in Figure~\ref{Fig:model}. The central cluster (served by
the BS) is indexed as the $0$-th cluster, whose users directly
communicate with the base station over relatively short distances.
Each of the remaining $M$ clusters is served by a half-duplexing
RS\footnote{In this architecture, the system design still has the
flexibility that each MS can be served by multiple RSs and BS: each
MS can be treated as multiple virtual MSs, each served by one RS. }.
Specifically, each RS forwards data packets from the base station to
users in the its cluster using the decode-and-forward (DaF)
strategy. The number of users in the $m$-th cluster is denoted as
$K_m$. For the notation convenience, we assume that the first $M$
users in the $0$-th cluster are the $M$ RSs, and the remaining
$K_0-M$ users in the $0$-th cluster are the MSs of the $0$-th
cluster ($K_0+K_1+...+K_M=K+M$).

The above secondary user (SU) system is assumed to opportunistically access a spectrum licensed to another network, whose users
are referred to as the \emph{primary users} (PU) and have the highest priority of using the spectrum. Primary users are
distributed over the service area of the SU system. To avoid interrupting the communication of primary users, every transmitter
(including the BS and the RSs) of the SU system is not allowed to transmit if there is active PU in the coverage.

The protocol for relay transmission is described as follows. The channels are assumed to be frequency selective and divided into
$N$ independent subchannels using the \emph{orthogonal frequency division multiplexing} (OFDM) modulation
\cite{GoldsmithBook:WirelessComm:05}. Downlink transmission is divided into frames, each with two phases (as illustrated in Figure \ref{Fig:frame}). In phase one, the base
station delivers packets to the MSs of the $0$-th cluster and all the RSs; in phase two, each RS forwards data packets to the MSs
in the corresponding cluster. To avoid interfering MSs in other clusters, we have the following assumption:
\begin{assumption}
\textcolor{black}{The base station does not deliver packets in phase
two. In order to control the inter-cluster interference between two
adjacent relay clusters, the transmitted signals at the adjacent RSs
are spread by different orthogonal spreading sequences in the
frequency domain as illustrated in Figure \ref{Fig:model}. }
\end{assumption}

\subsection{\textcolor{black}{Channel Model}}
\textcolor{black}{The channel realization is assumed to be
quasi-static over one frame but independent and identically
distributed (i.i.d.) across different frames. Channel gains are
characterized by the long-term path loss, shadowing and the
short-term fading. The symbol received at the $k$-th user of the
$m$-th cluster in the $n$-th subchannel, denoted as $Y_{m,n,k}$, can
be written as}
\begin{equation}
Y_{m,n,k} = \sqrt{p_{m,n,k}l_{m,k}} H_{m,n,k} X_{m,n,k} + Z_{m,n,k}, \nonumber
\end{equation}
where $X_{m,n,k}$ is the transmitted symbol, $p_{m,n,k}$ is the
transmission power, \textcolor{black}{$l_{m,k}$ is the long-term
channel attenuation due to path loss and shadowing, $H_{m,n,k}\sim
\mathcal{CN}(0,1)$ models short-term fading,} and $Z_{m,n,k} \sim
\mathcal{CN}(0,1)$ represents the additive white Gaussian noise.
Note that $H_{m,n,k}$ represents the channel between the $k$-th user
and the base station if $m=0$, or the $m$-th relay station if $m>0$.

The BS and RSs adapt the data rates, power, subchannel allocation
for the downlink transmission based on the CSI at the transmitter
(CSIT). We consider a time division duplex (TDD) system where the
CSIT can be acquired by channel reciprocal
\cite{MarzHoch:FastTransferCSI:2006}. Due to CSI estimation noise as
well as duplexing delay, the CSIT obtained will not be accurate and
the CSIT error model (based on MMSE prediction) is given by
\cite{MarzHoch:FastTransferCSI:2006}:
\begin{equation}
\hat{H}_{m,n,k} = H_{m,n,k} + \Delta H_{m,n,k},\quad \forall m,n,k
\end{equation}
where  $H_{m,n,k}$ represents actual CSI, $\Delta H_{m,n,k}$
represents the CSIT error which is modelled as complex Gaussian
distribution with mean 0 and variance $\sigma^2_e$ ($\Delta
H_{m,n,k}\sim\mathcal{CN}(0,\sigma^2_e)$), and $\mathbf{E}[\Delta
H_{m,n,k} \hat{H}_{m,n,k}]=0$ (meaning that the estimation error
$\Delta H_{m,n,k}$ is uncorrelated to CSIT $\hat{H}_{m,n,k}$). For
convenience, the CSIT is grouped according to cluster as the sets
$\mathbf{\hat{H}}_m = \cup_{n,k}\{\hat{H}_{m,n,k} \}$ for $0\leq m
\leq M$, which are referred to as \emph{local CSIT} at the $m$-th
cluster. The set $\mathbf{\hat{H}} = \bigcup\limits_{m=0}^M
\mathbf{\hat{H}}_m$ is called as \emph{global CSIT}.

\subsection{Dynamic Spectrum Access and Fairness Issues}

In each cluster, each secondary user senses the spectrum and
searches for subchannels unused by primary users, which, for
instance, may be wireless microphones or other Part 74 devices
\cite{SunminLim:07}. The spectrum sensing results consist of binary
indicators specifying the availability of subchannels. These sensing
resutls are referred to as \emph{raw sensing information} (RSI) in
this paper. Let $\hat{S}_{m,n,k}\in \{0,1\}$ denote the sensed state
at the $k$-th user on the $n$-th subchannel in the $m$-th cluster,
where $\hat{S}_{m,n,k} = 1$ and $0$ correspond to the states
``available" and ``unavailable", respectively. RSI
$\hat{\mathbf{S}}_m=\{\hat{S}_{m,n,k}|\forall n,k\}$ is communicated
by users to their corresponding servers (BS/RS) for enabling
resource allocation. Moreover, we also define the aggregation of RSI
from all clusters as $\hat{\mathbf{S}}=\{\hat{\mathbf{S}}_m |\forall
m\}$. Let $S_{m,n}$ be the actual primary-user state on the $n$-th
subchannel in the $m$-th cluster with $S_{m,n}=1$ denoting
subchannel is actually available and $S_{m,n}=0$ denoting otherwise,
$\mathbf{S}_m = \{S_{m,n}|\forall n\}$ be the actual PU activity of
all the subchannels in the $m$-th cluster
 and $\mathbf{S}=\{\mathbf{S}_m|\forall m\}$ be the aggregation of actual PU activity of all
clusters which is quasi-static over a number of frames\footnote{In practice, the PU activity changes over a longer time scale
compared with the CSI.}. Moreover, define $q_p = \Pr(S_{m,n} = 1)$ as the probability one subchannel is available, which is assumed
to be identical for all $m$ and $n$. In practice, we cannot have perfect sensing at the mobile and there exist nonzero
probabilities for the events \emph{false alarm} $\left(q_f = Pr(\hat{S}_{m,n,k} = 0|S_{m,n}=1)\right)$ and \emph{mis-detection}
$\left(q_m = Pr(\hat{S}_{m,n,k} = 1|S_{m,n}=0)\right)$ \cite{Zhao:SurveyDynamSpectAccess:2007}. Moreover, $q_d= 1-q_m$ represents
the probability of detection.

Due to the imperfect sensing measurement, it is not possible to
eliminate the interference from the SU to the PU systems. To protect
communication in the PU networks, we require
\begin{equation}
I_{m,n}=(\sum_{k=1}^{K_m}p_{m,n,k})\tau_{m,n}(1-\mathbf{E}[S_{m,n}|\widehat{\bf
S}_{m,n}])\leq \overline{I}, \quad \forall m,n,\label{eqn:inf-cont}
\end{equation}
where $I_{m,n}$ is the conditional average interference level (conditioned on the sensing measurement) from the SU (at
the $m$-th cluster and the $n$-th subchannel) to the active PU, $ \widehat{\bf S}_{m,n}=\{{\hat S}_{m,n,k} | k\in \{1
,K_m\}\} $, $p_{m,n,k}$ is the transmit power of the $m$-th RS (or BS) to its $k$-th MS in the $n$-th subchannel,
$\tau_{m,n}$ is the path loss between the SU transmitter (at the $m$-th cluster and the $n$-th subchannel) and the active
PU. Thus, each SU transmitter should guarantee that the average interference to the active PU in its cluster area is
not larger than one tolerance threshold $\overline{I}$.

{\bf Remarks (Fairness Issue with Cognitive OFDMA Systems without
RS):} \textcolor{black}{Consider a simple scenario where we have one
PU in each of the M RS clusters as well as the BS cluster as shown
in Figure \ref{Fig:model}. As a result, there are $M+1$ PUs in the
system.  Let $q_p$ be the probability that the PU in a cluster
becomes active in one subchannel. If there are no RS in the SU
system in Figure \ref{Fig:model}, the access opportunity of a
cell-edge user (users in the cluster $m>0$) in one subchannel is
$(1-q_p)^{M+1}$, which is the probability for all the $M+1$ PUs in
the BS's coverage area to be idle. Hence, the cell-edge users could
hardly access the spectrum even for moderate PU activity, leading to
critical fairness issue.}

\section{Joint Control of Rate, Power and Subchannel Allocation: Problem Formulation}\label{Section:ProbForm}

In this section, we shall formulate the rate, power and subchannel allocation design as an optimization problem. We first
formally define the optimization variables (control policies) as well as the optimization objectives below.

\subsection{Definitions of Control Policies}\label{Section:DefPolicy}
Consider  transmitting to the $k$-th user in the $0$-th cluster (the BS's cluster) over the $n$-th subchannel. The transmission power, rate
and percentage of subchannels the base station allocates to the user is denoted as $p_{0,n,k}(\hat{\mathbf{H}},
\hat{\mathbf{S}})$, $r_{0,n,k}(\hat{\mathbf{H}}, \hat{\mathbf{S}})$ and $\alpha_{0,n,k}(\hat{\mathbf{H}},
\hat{\mathbf{S}})$ respectively, which are adapted to the imperfect CSIT $\hat{\mathbf{H}}$ and RSI $\hat{\mathbf{S}}$.
The corresponding polices for controlling transmit power ($\mathcal{P}_0$), subchannel allocation ($\mathcal{A}_0$) and
transmit data rate ($\mathcal{R}_0$) are defined as the function sets
$\mathcal{P}_0:=\left\{p_{0,n,k}(\hat{\mathbf{H}}, \hat{\mathbf{S}})| \forall n,k\right\}$,
$\mathcal{A}_0=\left\{\alpha_{0,n,k}(\hat{\mathbf{H}},\hat{\mathbf{S}})| \forall n,k\right\}$, and
$\mathcal{R}_0=\left\{r_{0,n,k}(\hat{\mathbf{H}},\hat{\mathbf{S}})| \forall n,k\right\}$. These policies must satisfy a
set of constraints. Specifically, assuming the total transmission power at the base station is fixed at $P_0$,
\begin{equation}\label{Eq:Const1}
\text{Power constraint (BS):}  \quad \sum_{n=1}^N\sum_{k=1}^{K_0} p_{0,n,k}(\hat{\mathbf{H}}, \hat{\mathbf{S}})\leq
P_0.
\end{equation}
By definition, the percentages of subchannels allocated to different
users/relay-stations satisfy
\begin{equation}\label{Eq:Const2}
\text{Subchannel allocation constraint (BS):}\quad \sum_{k=1}^{K_0}
\alpha_{0,n,k}(\hat{\mathbf{H}},\hat{\mathbf{S}})\leq 1, \quad
\forall n\in[1,N].
\end{equation}
Furthermore, the data rates are adjusted under a constraint on the per-hop packet error probability\footnote{We assume
sufficiently strong coding, such as LDPC, is used so that the PER is dominated by the channel outage (transmit data
rate less than the instantaneous mutual information). This is reasonable as it has been shown \cite{Chung:2000} that
LDPC for reasonable block size (e.g. 8kbyte) could achieve the Shannon's limit to within 0.05dB. } $\Pout$, namely that
for given a per-hop PER constraint $0< \epsilon < 1$
\begin{equation}\label{Eq:Const3}
\text{Per-hop outage constraint (BS):}\quad \Pout(r_{0,n,k}, \hat{\mathbf{H}}) = \Pr (r_{0,n,k} >
\mathfrak{R}_{0,n,k}|\hat{\mathbf{H}}) = \epsilon, \quad \forall n\in[1,N],k\in[1,K_0],
\end{equation}
where $\mathfrak{R}_{0,n,k}$ is the maximum achievable data rate from the base station to $k$-th user in the $n$-th subchannel.

Each packet transmitting from the base station to a relay station is designed to contain information bits for users to
be served by this RS in the cluster. Let $d_{m,n,k}$ be the fraction of $k$-th user's information bits in a packet
transmitted over the $n$-th subchannel and received at the $m$-th relay station. It follows from the definition that
\begin{equation}
\text{Packet partition constraint (BS):}\quad     \sum_{k=1}^{K_m} d_{m,n,k}\leq 1, \ \forall m>0,n.
\end{equation}
 The base station is assumed to control $\{d_{m,n,k}\}$ based on the CSIT and RSI. The corresponding control policy for the
$m$-th RS is defined as $ \mathcal{D}_m:=\left\{d_{m,n,k}(\hat{\mathbf{H}},\hat{\mathbf{S}})|\forall n,k\right\}$.
Moreover, we also define the system packet partition policy as $\mathcal{D}= \bigcup\limits_{m=1}^M \mathcal{D}_m$.

The policies used by a relay station depend on the packet receiving status of the phase one transmission. Let $t_{n,m}\in
\{0,1\}$ denote the indicator of the decoding state of the $m$-th relay station on the $n$-th subchannel, where $t_{n,m}=1$ means
the corresponding packet is decoded successfully and $t_{n,m}=0$ means otherwise. Moreover, define the set
$\mathbf{T}_{m}=\{t_{n,m}|\forall n \in [1,N]\}$. Adding the newly defined sets as input, the policies for controlling power,
rate, and subchannel allocation at relay stations are defined similarly to those for the base station as $\mathcal{P}_m
:=\left\{p_{m,n,k}(\hat{\mathbf{H}},\hat{\mathbf{S}},\mathbf{T}_m) | \forall n,k\right\}$, $\mathcal{A}_m
:=\left\{\alpha_{m,n,k}(\hat{\mathbf{H}},\hat{\mathbf{S}},\mathbf{T}_m) | \forall n,k\right\}$, and
$\mathcal{R}_m=\left\{r_{m,n,k}(\hat{\mathbf{H}},\hat{\mathbf{S}},\mathbf{T}_m)| \forall n,k\right\}$.
 These policies must satisfy the following constraints
\begin{eqnarray}
\text{Power constraint (relay):}&&\sum_{n=1}^N\sum_{k=1}^{K_m}
p_{m,n,k}\leq P_m, \forall m \in [1,M]\label{Eq:Const1:r}
\end{eqnarray}
\begin{eqnarray}
\text{Subchannel allocation constraint (relay):}&& \sum_{k=1}^{K_m}\alpha_{m,n,k} \leq 1, \  \forall m \in [1,M],n \in [1,N]
\label{Eq:Const2:r}
\end{eqnarray}
\begin{eqnarray}
\text{Per-hop outage constraint (relay):}&&\Pout(\{r_{m,n,k}\}, \hat{\mathbf{H}}, \mathbf{T}_m) =
\Pr(r_{m,n,k}>\mathfrak{R}_{m,n,k}|\hat{\mathbf{H}}) =
\epsilon\label{Eq:Const3:r}
\end{eqnarray}
\textcolor{black}{
\begin{eqnarray}
\text{Flow balance constraint:}&&
 \sum_{n=1}^N r_{m,n,k} \leq  \sum_{n=1}^N
  d_{m,n,k}t_{n,m}r_{0,n,m}, \ \forall
  m\in[1,M],k \in [1,K_m].\label{Eq:FlowConst}
\end{eqnarray}}
where $\mathfrak{R}_{m,n,k}$ is the maximum achievable data rate from the $m$-th relay station to $k$-th user in the $n$-th
subchannel, the last constraint (\ref{Eq:FlowConst}) is because the total information bits transmitted by each RS cannot be more
than the information bits received from the BS.

\subsection{Average Weighted Goodput and Fairness}
The average weighted goodput is defined and used in the sequel as the metric for optimizing control policies discussed
in the preceding section. When the PU is not active at the $m$-th cluster and the $n$-th subchannel ($S_{m,n}=1$), the
instantaneous mutual information between the $m$-th transmitter and the $k$-th receiver in the $n$-th subchannel is given
by:
\begin{equation}
C_{m,n,k} =g_m \alpha_{m,n,k}\log_2\left(1+\frac{p_{m,n,k}l_{m,k}|H_{m,n,k}|^2}{\alpha_{m,n,k}}\right) \mbox{ if }
S_{m,n} = 1, \nonumber
\end{equation}
where $g_m \in \{0.25,0.5\}$ ($g_0=0.5$ and $g_m=0.25$ $\forall m>0$) is a constant indicating the spectrum efficiency. Due to
the half-duplex constraint at the base station, $g_m$ is equal to $0.5$ for $m=0$ (base station's cluster). Moreover, due to the
half-duplex constraint and the orthogonal spreading at the RSs, $g_m = 0.25$ for $m \geq 1$ (relay stations' clusters). On the
other hand, we have the following assumption on the interference from PU to SU:
\begin{assumption}
We assume the power of active PU is large, so that the SU transmission in one cluster will fail if there is any active PU in that
cluster using the same subchannel.
\end{assumption}
Hence, when $S_{m,n} = 0$ (PU active), there is large interference from the PU and the instantaneous mutual information can be
regarded as $C_{m,n,k} = 0$. As a result, the instantaneous mutual information can be written as:
\begin{equation}
C_{m,n,k} = \left\{
              \begin{array}{ll}
               \underbrace{g_m \alpha_{m,n,k}\log_2\left(1+\frac{p_{m,n,k}l_{m,k}|H_{m,n,k}|^2}{\alpha_{m,n,k}}\right)}_{\mbox{$\mathfrak{R}_{m,n,k}$}}, & \hbox{if $S_{m,n}=1$;} \\
               0, & \hbox{if $S_{m,n}=0$.}
              \end{array}
            \right. \nonumber
\end{equation}
Due to the imperfect CSIT knowledge, there is uncertainty in the instantaneous mutual information $C_{m,n,k}$ at the
transmitters and hence, there will be potential packet errors due to channel outage if the scheduled data rate exceeds
$C_{m,n,k}$. This packet error is {\em systematic} and cannot be alleviated by using strong error correction coding. As
a result, we shall consider goodput (b/s successfully delivered to the mobiles) as our performance measure.  The
instantaneous goodput over the $(m,n,k)$-th subchannel is defined as
\begin{eqnarray}
U_{m,n,k} &=& r_{m,n,k}\mathbf{I}(r_{m,n,k}\leq C_{m,n,k}) \nonumber\\
&=& r_{m,n,k} S_{m,n} \mathbf{I}(r_{m,n,k} \leq \mathfrak{R}_{m,n,k}) \nonumber,
\end{eqnarray}
where $\mathbf{I}(A)$ is the indicator function
with value $1$ when the event $A$ is true and $0$ otherwise.

Let $\{w_{m,k}\}$ be a set of goodput weights for different users (the weight for the $k$-th user in the $m$-th cluster is
$w_{m,k}$), whose values are set according to the users' QoS priorities. The average weighted goodput is given below:
\begin{eqnarray}
\overline{G}({\cal A}, {\cal P},  {\cal D})&:=& \mathbf{E}_{\mathbf{S},\hat{\mathbf{S}},\mathbf{H},\hat{\mathbf{H}}}\bigg[\sum_{n=1}^N\sum_{k=M+1}^{K_0}w_{0,k}U_{0,n,k} + \sum_{m=1}^M\sum_{n=1}^N\sum_{k=1}^{K_m}w_{m,k}U_{m,n,k}\bigg]\nonumber\\
&=&\mathbf{E}_{\hat{\mathbf{S}},\hat{\mathbf{H}}}\underbrace{\left\{\mathbf{E}_{\mathbf{S},\mathbf{H}}
\bigg[\sum_{n=1}^N\sum_{k=M+1}^{K_0} w_{0,k}U_{0,n,k}\mid \hat{\mathbf{S}},\hat{\mathbf{H}}\bigg]+ \sum_{m=1}^M
\mathbf{E}_{\mathbf{S},\mathbf{H}}\left[ \sum_{n=1}^N\sum_{k=1}^{K_m}w_{m,k}U_{m,n,k}\mid
\hat{\mathbf{S}},\hat{\mathbf{H}}\right]\right\}}_{\widetilde{G}(\mathbf{ A}, \mathbf{P},
\mathbf{D}|\hat{\mathbf{S}},\hat{\mathbf{H}})}\nonumber,
\end{eqnarray}
where $\widetilde{G}$ defined above is referred to as the conditional average system goodput (conditioned on
$\hat{\mathbf{S}},\hat{\mathbf{H}}$), ${\cal A}=\{{\cal A}_m | \forall m \in [0,M]\}$, ${\cal P}=\{{\cal P}_m| \forall
m \in [0,M]\}$ and ${\cal D}=\{{\cal D}_m | \forall m \in [1,M]\}$  are the subchannel allocation policy, power allocation
policy and packet partition policy of the system respectively, $\mathbf{A}$, $\mathbf{P}$ and $\mathbf{D}$ denote the
subchannel allocation action, power allocation action and packet partition action of the system respectively for a given
global CSIT $\hat{\mathbf{S}}$ and global RSI $\hat{\mathbf{H}}$ .

{\bf Remarks (Incorporating Fairness in the weighted Goodput):} Note that the optimization objective in the above equation
embraces fairness in the resource allocation. For instance, users with higher priorities could be allocated a larger weight
$w_{m,k}$. Furthermore, proportional fair scheduling (PFS), which is a commonly used fairness attribute, is also embraced by
setting $w_{m,k}(t) = \frac{1}{{\widetilde R}_{m,k}(t)}$, where $w_{m,k}(t)$ is the weight of the $k$-th users at the $m$-th
cluster and $t$-th frame and ${\widetilde R}_{m,k}(t)$ is the measured average throughput of this user. ${\widetilde R}_{m,k}(t)$
is updated on each frame according to ${\widetilde R}_{m,k}(t) = (1-\frac{1}{t_s}){\widetilde
R}_{m,k}(t-1)+\frac{1}{t_s}\sum_{n=1}^N r_{m,n,k}(t)$, where $t_s$ is the duration of one frame and $r_{m,n,k}(t)$ is the
scheduled data rate of the user in the $t$-th frame.

Notice that
\begin{eqnarray}
\mathbf{E}_{\mathbf{S},\mathbf{H}}[U_{0,n,k} | \hat{\mathbf{S}},\hat{\mathbf{H}}] &=& \beta_{0,n} r_{0,n,k}
(1-\Pr[r_{0,n,k}>\mathfrak{R}_{0,n,k}|\hat{\mathbf{H}}]) \nonumber \\
&=& \beta_{0,n} r_{0,n,k}
(1-\epsilon) \nonumber\\
\mathbf{E}_{\mathbf{S},\mathbf{H}}\left[ \sum_{n=1}^N\sum_{k=1}^{K_m}w_{m,k}U_{m,n,k}\mid
\hat{\mathbf{S}},\hat{\mathbf{H}}\right] &=& \mathbf{E}_{\mathbf{T}_m,\mathbf{S}_m,\mathbf{H}_m} \left[
\sum_{n=1}^N\sum_{k=1}^{K_m}w_{m,k}U_{m,n,k}\mid \hat{\mathbf{S}},\hat{\mathbf{H}}\right] \nonumber \\
&=& \mathbf{E}_{\mathbf{T}_m}  \left[ \sum_{n=1}^N\sum_{k=1}^{K_m} \mathbf{E}_{\mathbf{S}_m,\mathbf{H}_m}
[w_{m,k}U_{m,n,k}|\mathbf{T}_m]\mid \hat{\mathbf{S}},\hat{\mathbf{H}}\right] \nonumber \\
&=& \mathbf{E}_{\mathbf{T}_m}\bigg[ \sum_{n=1}^N\sum_{k=1}^{K_m}w_{m,k}\beta_{m,n}r_{m,n,k}\big(1-\Pr[r_{m,n,k}>
  \mathfrak{R}_{m,n,k}|\hat{\mathbf{H}}]\big)\bigg] \nonumber \\
&=& \mathbf{E}_{\mathbf{T}_m}\bigg[ \sum_{n=1}^N\sum_{k=1}^{K_m}w_{m,k}\beta_{m,n}r_{m,n,k}\big(1-\epsilon\big)\bigg]
\nonumber
\end{eqnarray}
where $\beta_{m,n} = \mathbf{E}[S_{m,n}|\hat{\mathbf{S}}]$ is the probability that the $n$-th subchannel in the $m$-th
cluster is available given the sensing feedbacks from the mobiles and
$\Pr[r_{m,n,k}>\mathfrak{R}_{m,n,k}|\hat{\mathbf{H}}]=\epsilon$ ($\forall m,n,k$) is conditional packet error
probability of one-hop link for given $\hat{\mathbf{H}}$, $\widetilde{G}$ can be written as
\begin{eqnarray}
\widetilde{G}(\mathbf{ A}, \mathbf{P},
\mathbf{D}|\hat{\mathbf{S}},\hat{\mathbf{H}})&=&\underbrace{\sum_{n=1}^N\sum_{k=M+1}^{K_0}w_{0,k}\beta_{0,n}r_{0,n,k}\big(1-\epsilon\big)}_{\widetilde{G}_0(\mathbf{A}_0,\mathbf{P}_0|\hat{\mathbf{S}},\hat{\mathbf{H}})}+
\sum_{m=1}^M \mathbf{E}_{\mathbf{T}_m}\bigg[
\underbrace{\sum_{n=1}^N\sum_{k=1}^{K_m}w_{m,k}\beta_{m,n}r_{m,n,k}\big(1-\epsilon\big)}_{\widetilde{G}_m(\mathbf{A}_0,\mathbf{P}_0,\mathbf{A}_m,\mathbf{D}_m,\mathbf{P}_m|\hat{\mathbf{S}}_m,\hat{\mathbf{H}}_m,\mathbf{T}_{m})}\bigg],\nonumber
\end{eqnarray}
or $$\widetilde{G}(\mathbf{ A}, \mathbf{P}, \mathbf{D}|\hat{\mathbf{S}},\hat{\mathbf{H}})=
\widetilde{G}_0(\mathbf{A}_0,\mathbf{P}_0|\hat{\mathbf{S}},\hat{\mathbf{H}}) + \sum_{m=1}^M \mathbf{E}_{\mathbf{T}_m}
\widetilde{G}_m(\mathbf{A}_0,\mathbf{P}_0,\mathbf{A}_m,\mathbf{D}_m,\mathbf{P}_m|\hat{\mathbf{S}}_m,\hat{\mathbf{H}}_m,\mathbf{T}_{m}),$$
where $\mathbf{A}_m=\{\alpha_{m,n,k}|\forall n,k\}$, $\mathbf{P}_m=\{p_{m,n,k}|\forall m,n,k\}$,
$\mathbf{D}_m=\{d_{m,n,k}|\forall n,k\}$ and $\mathbf{D} = \{\mathbf{D}_m|\forall m\}$, $\mathbf{A} =
\{\mathbf{A}_m|\forall m\}$, $\mathbf{P} = \{\mathbf{P}_m|\forall m\}$.

\subsection{Problem Formulation}
Since a policy consists of a set of actions for each realization
of CSIT and RSI, finding the optimal policy is equivalent to the
following problem.
\begin{Problem}For each given CSIT $\hat{\mathbf{H}}$ and RSI $\hat{\mathbf{S}}$ realization, we have:
\begin{eqnarray}
&&\{\mathbf{A}^*(\hat{\mathbf{H}},\hat{\mathbf{S}}),
\mathbf{P}^*(\hat{\mathbf{H}},\hat{\mathbf{S}}), \mathbf{D}^*(\hat{\mathbf{H}},\hat{\mathbf{S}})\}\nonumber\\
&=&\max_{\mathbf{A}_0,\mathbf{P}_0,\mathbf{D}}\bigg\{\widetilde{G}_0(\mathbf{A}_0,\mathbf{P}_0|\hat{\mathbf{S}},\hat{\mathbf{H}})
  +\sum_{m=1}^M \mathbf{E}_{\mathbf{T}_{m}}\left[ \underbrace{\max_{\mathbf{A}_m,\mathbf{P}_m} \widetilde{G}_m(\mathbf{A}_0,\mathbf{P}_0,\mathbf{A}_m,\mathbf{D}_m,\mathbf{P}_m|\hat{\mathbf{S}}_m,\hat{\mathbf{H}}_m,\mathbf{T}_{m})}_{\mbox{Local Optimization on } \widetilde{G}_m}\right]\bigg\}\nn\\
&s.t.& \text{the constraints in \eqref{eqn:inf-cont}-\eqref{Eq:FlowConst},}\nonumber
\end{eqnarray}
\label{prob:main}
\end{Problem}

{\bf Remarks (Comparison with Traditional Resource Allocation Problem in OFDMA Systems):} Noting that neither the
objective function nor the constraint \eqref{Eq:FlowConst} is convex, the traditional optimization approaches in
\cite{WooseokNam:07,Oyman:07} cannot be applied in this problem as the duality gap is not zero. Moreover, due to the
potential packet error at the BS-RS link, the traditional centralized controller needs to solve $\mathcal{O}(M2^N)$
control variables for all possible $\mathbf{T}_m$ realization in Problem \ref{prob:main} (RS's control actions are the
function of $\mathbf{T}_m$). Thus, the brute force solution for Problem \ref{prob:main} involves unacceptable
computational complexity and huge communication overhead between the BS and the RSs. In this paper, we shall show how
to {\em divide and conquer} this non-convex optimization problem into the optimization problem at the BS and RSs. By
appropriate design of {\em backward recursion} and {\em online strategy}, the system only need to solve
$\mathcal{O}(MN)$ control variables. Furthermore, the algorithm can be implemented distributively in the system and the
communication overhead between the BS and RSs is very small.

\textcolor{black}{In Problem \ref{prob:main}, the local optimization
on $ \widetilde{G}_m $  with respect to $ \mathbf{A}_m $ and $
\mathbf{P}_m $ ($ \max\limits_{\mathbf{A}_m,\mathbf{P}_m}
\widetilde{G}_m (\cdot) $) is subject to the constraints
\eqref{eqn:inf-cont} ($m>0$) and
\eqref{Eq:Const1:r}-\eqref{Eq:FlowConst}. As a result, for a given
Phase-I receiving status $ \{r_{0,n,m}t_{n,m}\} $ and the packet
partitioning $\{d_{m,n,k}\}$, these local optimizations on
$\max\limits_{\mathbf{A}_m,\mathbf{P}_m} \widetilde{G}_m(\cdot)$ can
be done locally at the m-th RS for $m\in{1,..,M}$. Therefore, using
standard argument of primal decomposition \cite{Palomar:07}, solving
Problem \ref{prob:main} is equivalent to solving the following two
subproblems:}

\begin{Subproblem}[Optimization at $m$-th RS]
\begin{eqnarray}
&&\widetilde{G}^*_m(\{r_{0,n,m}t_{n,m}\},\{d_{m,n,k}\}|\hat{\mathbf{S}}_m,\hat{\mathbf{H}}_m)=\max_{\mathbf{A}_m,\mathbf{P}_m} \widetilde{G}_m
(\mathbf{A}_m,\mathbf{P}_m,\{r_{0,n,m}t_{n,m}\},\{d_{m,n,k}\}|\hat{\mathbf{S}}_m,\hat{\mathbf{H}}_m)\nonumber\\
&=&\max_{\mathbf{A}_m,\mathbf{P}_m}
  \sum_{n=1}^N\sum_{k=1}^{K_m}\frac{1}{4}(1-\epsilon)w_{m,k}\beta_{m,n}\alpha_{m,n,k}\log_2(1+\frac{p_{m,n,k}}{\alpha_{m,n,k}}\phi_{m,n,k})\nonumber\\
&s.t.& \text{the constraints in \eqref{eqn:inf-cont} ($m>0$), \eqref{Eq:Const1:r}-\eqref{Eq:FlowConst}.}\nonumber
\end{eqnarray}
where $\phi_{m,n,k}=\frac{1}{2}l_{m,n}\sigma_e^2F_{|\hat{H}_{m,n,k}|^{2}/\frac{1}{2}\sigma_e^2}^{-1}(\epsilon)$ which
is obtained from the outage probability constraint, and $F_{|\hat{H}_{m,n,k}|^{2}/\frac{1}{2}\sigma_e^2}^{-1}(\cdot)$
denotes the inverse cdf of non-central chi-square random variable with 2 degrees of freedom and non-centrality
parameter $|\hat{H}_{m,n,k}|^{2}/\frac{1}{2}\sigma_e^2$.\label{sp:rs}
\end{Subproblem}
\begin{Subproblem}[Optimization at the BS]
\begin{eqnarray}
&\max\limits_{\mathbf{A}_0,\mathbf{P}_0,\mathbf{D}}
&\widetilde{G}_0(\mathbf{A}_0,\mathbf{P}_0|\hat{\mathbf{S}},\hat{\mathbf{H}})
  +\sum_{m=1}^M \mathbf{E}_{\mathbf{T}_{m}} \widetilde{G}^*_m(\{r_{0,n,m}t_{n,m}\},\{d_{m,n,k}\}|\hat{\mathbf{S}}_m,\hat{\mathbf{H}}_m)\nonumber\\
=&\max\limits_{\mathbf{A}_0,\mathbf{P}_0,\mathbf{D}}
&\sum_{n=1}^N\sum_{k=M+1}^{K_0}\frac{1}{2}(1-\epsilon)w_{0,k}\beta_{0,n}\alpha_{0,n,k}\log_2(1+\frac{p_{0,n,k}}{\alpha_{0,n,k}}\phi_{0,n,k})\nonumber\\
& &+\sum_{m=1}^M \mathbf{E}_{\mathbf{T}_{m}}
\widetilde{G}^*_m(\{r_{0,n,m}t_{n,m}\},\{d_{m,n,k}\}|\hat{\mathbf{S}}_m,\hat{\mathbf{H}}_m)\nonumber\\
=&\max\limits_{\mathbf{A}_0,\mathbf{P}_0}
&\sum_{n=1}^N\sum_{k=M+1}^{K_0}\frac{1}{2}(1-\epsilon)w_{0,k}\beta_{0,n}\alpha_{0,n,k}\log_2(1+\frac{p_{0,n,k}}{\alpha_{0,n,k}}\phi_{0,n,k})\nonumber\\
& &+\sum_{m=1}^M \mathbf{E}_{\mathbf{T}_{m}}
\widetilde{G}^{**}_m(\sum_n r_{0,n,m}t_{n,m}|\hat{\mathbf{S}}_m,\hat{\mathbf{H}}_m)\label{eq:CAWSG}\\
&s.t.& \text{the constraints in \eqref{eqn:inf-cont} ($m=0$),  and \eqref{Eq:Const1}-\eqref{Eq:Const3},}\nonumber
\end{eqnarray}\label{sp:bs}
where
\begin{equation}
\widetilde{G}^{**}_m(\sum_n r_{0,n,m}t_{n,m}|\hat{\mathbf{S}}_m,\hat{\mathbf{H}}_m) = \max_{\{d_{m,n,k}|\forall k\}}
  \widetilde{G}^*_m(\{r_{0,n,m}t_{n,m}\},\{d_{m,n,k}\}|\hat{\mathbf{S}}_m,\hat{\mathbf{H}}_m) \nonumber.
\end{equation}
\end{Subproblem}

The divide and conquer procedure to solve Problem \ref{prob:main} is given below:
\begin{itemize}
\item {\bf Backward Recursion}: At the beginning of each frame, after channel estimation and sensing,
each RS calculates and feedbacks the function $\tilde{G}_m^{**}(r|\hat{\mathbf{S}}_m,\hat{\mathbf{H}}_m)$ to the BS.

\item {\bf Online Strategy}: In phase one, the BS solves the Subproblem \ref{sp:bs} and delivers packets accordingly. In
phase two, each RS (say the $m$-th RS) solves its Subproblem \ref{sp:rs} according to the packet receiving status in
phase one $\{r_{0,n,m}t_{n,m}|\forall n=1,2,...,N\}$, and delivers packets accordingly.
\end{itemize}

\section{Joint Control of Rate, Power and Subchannel Allocation: Solutions}\label{Section:Solution}

In this section, We shall derive a low-complexity solution for the general weighted goodput optimization. The solution supports
decentralized implementation which significantly reduce computational complexity and signaling loading. Furthermore, the solution
is asymptotically optimal when the number of users is sufficiently large and the BS-RS links are sufficiently good. We shall also derive the solution for PFS as a special case.

\subsection{Asymptotically Optimal Algorithm}

\textbf{Solution of Subproblem 1:} The Subproblem 1 can be solved
by using the  duality approach \cite{BoydBook}. Specifically, the
Lagrangian is given as
\begin{eqnarray}
L_m&=&\widetilde{G}_m-\sum_{n=1}^N\lambda_n\left(\sum_{k=1}^{K_m}
\alpha_{m,n,k} -1\right) -\nu\left(\sum_{n=1}^N
\sum_{k=1}^{K_m} p_{m,n,k} - P_m\right)\nonumber\\
& &-\sum_{n=1}^N \eta_n\left(\sum_{k=1}^{K_m} (1-\beta_{m,n}) \tau^2_{m,n} p_{m,n,k} -
\overline{I}\right)-\sum_{k=1}^{K_m}\mu_k\left(\sum_{n=1}^{N}\frac{\alpha_{m,n,k}}{4}\log_2(1+\frac{p_{m,n,k}\phi_{m,n,k}}{\alpha_{m,n,k}})-R_{m,k}\right)\nonumber
\end{eqnarray}
where $R_{m,k}=\sum_{n=1}^N t_{n,m} d_{m,n,k} r_{0,n,m}$ is constant in this subproblem. Hence, the dual problem is:
\begin{Subproblem}[Dual Problem of Subproblem 1]
\begin{eqnarray}
&\min\limits_{\vec{\lambda},\vec{\eta},\vec{\mu},\nu}&
\max_{\mathbf{A}_m,\mathbf{P}_m}L_m(\vec{\lambda},\vec{\eta},\vec{\mu},\nu)
\nonumber\\
&s.t.& \vec{\lambda},\vec{\eta},\vec{\mu},\nu \succeq 0, \nonumber
\end{eqnarray}
where $\mathbf{A}\succeq0$ means each element of vector
$\mathbf{A}$ is nonnegative.
\end{Subproblem}
The algorithm to solve the above dual problem is presented in
Appendix A. Note that the Subproblem 1 is a non-convex
optimization problem because the optimization constraint is
non-convex. Nevertheless, since the problem satisfies the property
of ``time sharing" as introduced in \cite{WeiYu:06}, the duality gap
of the above problem is zero, and hence, solving the above dual
problem will lead to the optimal solution of Subproblem 1.

\textbf{Solution of Subproblem \ref{sp:bs}:} The expectation on the binary vector $\mathbf{T}_m$ in Subproblem
\ref{sp:bs} should take over exponential order (w.r.t. the number of subchannels $N$) of possible situations, which raises
unacceptable computational complexity. In the following lemma, we show that the expectation over the binary vector
$\mathbf{T}_m$ can be decoupled into each subchannel asymptotically, therefore, the computational complexity become
linear.
\begin{lemma}[Asymptotically Equivalent Objective]
When the channels between the BS and the RSs are sufficiently good, one relay is scheduled at most on one subchannel.
Hence, (\ref{eq:CAWSG}) can be written as
\begin{eqnarray}
&&\max\limits_{\mathbf{A}_0,\mathbf{P}_0}
  \widetilde{G}_0(\mathbf{A}_0,\mathbf{P}_0|\hat{\mathbf{S}}_m,\hat{\mathbf{H}}_m)
+\sum_{m=1}^M  \sum_{n=1}^N (1-\epsilon)\beta_{0,n}
\widetilde{G}_{m}^{**}(r_{0,n,m}|\hat{\mathbf{S}}_m,\hat{\mathbf{H}}_m)\nonumber.
\end{eqnarray}\label{lem:obj}
\end{lemma}
\begin{proof}
Please refer to Appendix B.
\end{proof}

Since the Subproblem \ref{sp:bs} is calculated at the BS, each RS should inform the expression of
$\widetilde{G}_{m}^{**}(r)$ to BS. The feedback of accurate $\widetilde{G}_{m}^{**}(r)$ expression involves large
feedback overhead. In the following lemma, we show that the feedback overhead can be significantly  reduced when the
user density $\rho$ is sufficiently large:
\begin{lemma}
When the user density $\rho$ is sufficiently large, $\widetilde{G}_{m}^{**}(r)$ is a convex piecewise linear
function.\label{lem:linear}
\end{lemma}
\begin{proof}
Please refer to Appendix C.
\end{proof}

The construction of function $\widetilde{G}_{m}^{**}(r)$ is presented in Appendix C as well. Moreover, an example of
function $\widetilde{G}_{m}^{**}(r)$ is illustrated in Figure \ref{Fig:linear}. With the conclusions of the Lemma
\ref{lem:obj}, the Subproblem \ref{sp:bs} is a convex optimization problem and can also be solved by the duality
approach (Similar to Subproblem \ref{sp:rs}) which is presented in in Appendix A. As a result, the overall decentralized resource
allocation algorithm for the relay-assisted CR system is summarized below:

\begin{Algorithm}[Decentralized Asymmetrical Optimal Control Algorithm]
The overall decentralized control algorithm includes the following steps:
\begin{itemize}
\item \textit{Step 1 (Cluster-Based Spectrum Sensing):} For $m=\{0,..,M\}$, mobiles in cluster $m$
deliver the 1-bit RSI to the cluster controller (BS or RS).

\item \textit{Step 2 (Backward Recursion):} The $m$-th RS feeds back the function
$\widetilde{G}_{m}^{**}(\cdot)$ to the BS.

\item \textit{Step 3 (Online Strategy --- Phase One):} From the local CSI ($\hat{\mathbf{H}}_0$),
local RSI ($\hat{\mathbf{S}}_0$) and $\widetilde{G}_{m}^{**}(\cdot)$, the BS determines the power, rate and subchannel
allocation of the mobiles in cluster $0$ as well as the RSs using the iterative algorithm for Subproblem $2$ in
Appendix A.

\item \textit{Step 4 (Online Strategy --- Phase Two):} If the $m$-th RS decodes the information
from the BS successfully, it will determine the power, rate, subchannel allocation to the MSs in its cluster based on
the local CSI ($\hat{\mathbf{H}}_m$) and RSI ($\hat{\mathbf{S}}_m$) using the solution of Subproblem $1$ in appendix A.

\end{itemize}
\end{Algorithm}

\textbf{Remarks:} The solution is decentralized in the sense that the computational loading is shared between the BS
and the RSs. Furthermore, only local CSI is needed at the $m$-th RS and the BS and this substantially reduces the
required signaling loading to deliver the global CSI in conventional centralized approach. While the $m$-th RS needs to
feedback $\widetilde{G}_{m}^{**}(\cdot)$ to the BS, the required signaling loading is very small because
$\widetilde{G}_{m}^{**}(\cdot)$ is a piecewise-linear function (as illustrated in Figure \ref{Fig:linear}) and it can
be characterized by O(ML) parameters in the worst case ($L$ is the number of QoS levels).

\subsection{PFS scheduling for Two-Hop RS-Assisted Cognitive OFDMA System}

The system objective function of PFS is given by $\sum\limits_{m,n,k} \frac{U_{m,n,k}}{\widetilde{R}_{m,k}}$, where
$\widetilde{R}_{m,k}$ is the average throughput of the $k$-th user in the $m$-th cluster. As a result, the PFS is a
special case of the weighted goodput objective considered in the paper. Yet, brute-force applications of the solution
in the pervious section in PFS will incur a large signaling overhead from the RS to the BS because the
$\widetilde{G}_{m}^{**}(\cdot)$ of PFS involves very large number of parameters (and hence, induce huge signaling
overhead for $m$-th RS to feedback $\widetilde{G}_{m}^{**}(\cdot)$ to the BS). In the following, we obtain a simple
characterization of $\widetilde{G}_{m}^{**}(\cdot)$ (which is asymptotically optimal) under PFS.

\begin{lemma}
Suppose the links between the base station and the relays are
sufficiently good, if $K_m$ is sufficiently large,
$\widetilde{G}_{m}^{**}(r)$ can be simplified as follows in
Subproblem 2
\begin{equation}
\widetilde{G}_{m}^{**}(r) = \left\{\begin{array}{ll}
\sum_{n=1}^N \frac{r w_{m,A_{m,n}}\beta_{m,n}(1-\epsilon)}{4R_m}\log_2(1+p_{m,n}l_{m,n,A_{m,n}}\phi_{m,n,A_{m,n}})&   r\leq R_m \\
\sum_{n=1}^N \frac{w_{m,A_{m,n}}\beta_{m,n}(1-\epsilon)}{4}\log_2(1+p_{m,n}l_{m,n,A_{m,n}}\phi_{m,n,A_{m,n}})  &  \mbox{Otherwise}\\
\end{array}\right. \label{eq:lem2}
\end{equation}
where $$R_m = \sum_{n=1}^N \frac{1}{4}\log_2(1+p_nl_{m,n,A_{m,n}}\phi_{m,n,A_{m,n}}),\quad p_{m,n} =
\frac{\beta_{m,n}P_0}{\sum_{n=1}^N\beta_{m,n}}$$ and $$A_{m,n} = \arg\max_k
w_{m,k}\log_2(1+p_{m,n}l_{m,n,k}\phi_{m,n,k}).$$\label{lem:pfs}
\end{lemma}

\begin{proof}
Please refer to Appendix D.
\end{proof}

Since $\widetilde{G}_{m}^{**}(r)$ can be parameterized by $\left(R_m,\sum_{n=1}^N
\frac{w_{m,A_{m,n}}\beta_{m,n}(1-\epsilon)}{4}\log_2(1+p_{m,n}l_{m,n,A_{m,n}}\phi_{m,n,A_{m,n}})\right)$, the feedback
overhead to deliver $\widetilde{G}_{m}^{**}(r)$ from the $m$-th RS to the BS is very small and does not scale with
$K_m$.

\section{Asymptotic Goodput of Two-Hop RS-Assisted Cognitive OFDMA Systems under PFS}\label{Section:AsympAna}

In this section, we analyze the asymptotic performance of the scheduling algorithm derived in the preceding section.
Specifically, the system throughput is derived for a sufficiently large number of users in each cluster. To obtain
insights on the performance gains, we impose a set of simplifying assumptions. We assume each cluster contains $K_c$
MSs. Furthermore, we assume line-of-sight link (with high gain antenna) between the RSs and the BS and hence, the
throughput is limited by the second hop. Finally, users will not be closer than $\gamma$ to the RS, where $\gamma$ is
certain fixed distance. The following theorem summarizes the asymptotic system goodput of the relay-assisted cognitive
OFDMA system under PFS.

\begin{theorem}
\textcolor{black}{Suppose there are $ M $ RS clusters and $ N $
independent subchannel in the system. Furthermore, consider a simple
scenario where there is one PU in each of the M RS clusters and the
BS cluster, as shown in Figure \ref{Fig:model}. Let $q_p$ be the
probability that a PU becomes active in a subchannel.  For
sufficiently large number of MSs per cluster $K_c$ and sufficiently
strong BS-RS links in the above system, the average throughput of
the $k$-th user in the $m$-th cluster ($m>1$) achieved under the
proportional fair scheduling is given by}
\begin{eqnarray}
\overline{T}_{m,k} &=& \frac{N(1-q_p)(1-q_p^N)}{K_c} \int_{0}^{+\infty}
  \frac{1}{4}\log_2 (1+\frac{P_m}{N} l_{m,k} x) d F_{max,K_c}(x) \\
&=& \frac{N(1-q_p)(1-q_p^N)}{4K_c}
  \log_2 (1+\frac{P_m}{N} l_{m,k} \ln K_c) \ \ \mbox{when }K_c \rightarrow +\infty,\label{eq:avg-relay}
\end{eqnarray}
where $F_{max,K_c}(x)$ is the CDF of $\max\{|H_{m,n,k}||\forall k\}$. The equivalent PFS scheduling rule at the RS is
given by
\begin{equation}
A_{m,n} = \arg\max_k \{|H_{m,n,k}||\forall k\} \label{eqn:asy-user}
\end{equation}
where $A_{m,n}$ is the selected user of the subchannel $n$ in
the $m$-th cluster.\label{thm:aym}
\end{theorem}

\begin{proof}
Please refer to Appendix E.
\end{proof}

\textcolor{black}{Using similar analysis as in Appendix E, it can be
shown that the average goodput (under PFS) of a mobile in a
cognitive OFDMA system without RS is given by:
\begin{eqnarray}
\overline{T}_{m,k}^{(b)} &=& (1-q_p)^{M+1}\frac{N}{MK_c} \log_2 (1+\frac{P_0}{N}l_{m,k}^{(b)}\ln
  K_cM)  \ \ \mbox{when }K_c \rightarrow +\infty \label{eq:avg-base},
\end{eqnarray}
where $F_{max,MK_c}(x)$ is the CDF of $\max\{|H_{m,n,k}||\forall k,m\}$, $l_{m,k}^{(b)}$ is the long-term path loss and shadowing from the user to the base station, the coefficient $(1-q_p)^{M+1}$ before the logarithm is the probability that one subchannel is available\footnote{Since $q_p$ is the probability that a PU will be active in one subchannel of one cluster, the probability that one subchannel being clean (no active PU) in the whole cell is $(1-q_p)^{M+1}$. Since the BS can transmit packets to the cell edge users in one subchannel only when this subchannel is clean in the whole cell area (i.e. all PUs in the coverage area are IDLE). Thus, the probability that one subchannel is available is $(1-q_p)^{M+1}$. }, the $N$ in the numerator of the coefficient before logarithm is because there are $N$ parallel independent subchannels, and the $MK_c$ in the denominator of the coefficient is because in each subchannel the access probability of each MS is $1/(MK_c)$.} Compared with the results in Theorem \ref{thm:aym}, it can be
concluded that
\begin{itemize}
\item The system goodput of the regular cognitive OFDMA system
without relay stations is $\mathcal{O}(N(1-q_p)^{M+1}\ln \ln MK_c)$, which is very sensitive to the PU activity due to
the factor $(1-q_p)^{M+1}$. For moderate $q_p$, the spectrum access opportunity of the cell-edge users is very small.

\item The spectrum access opportunity of the cell-edge users can
be improved by employing relays. Active primary users in one relay cluster would not affect the packet transmission on
other relay clusters as illustrated by the factor $(1-q_p)(1-q_p^N)$ in equation (\ref{eq:avg-relay}). Moreover, the
receiving SNR at the mobile users is significantly increased by employing relays ($l_{m,k}>>l_{m,k}^{(b)}$). As a
result, the relay-assisted CR system can achieve much higher system throughput than the baseline system without relays
under PFS.
\end{itemize}

\section{Simulation Results and Discussions}\label{Section:Sim}
In this section, we shall compare the performance of the proposed
relay-assisted cognitive OFDMA system with several baseline systems.
Baseline 0 refers to a naive design of a cognitive OFDMA system
(without RS) where the power, rate and subchannel allocation are
designed assuming perfect CSIT. \textcolor{black}{Baseline 1 refers
to the Separate and Sequential Allocation (SSA) in relay-assisted
cognitive OFDMA system, which is a semi-distributed scheme proposed
for relay-assisted OFDMA systems in \cite{RS:10}. Similar approach
also appears in \cite{Awad:08,ZhihuaTang:09}.} Baseline 2 and 3
refer to a similar cognitive OFDMA system (without RS). Moreover, in
baseline 1 and 2, the PU activity in RS clusters is the same as that
in BS cluster, i.e. $q_{pm}=q_{p0}$. In baseline 3, the PU activity
in RS clusters is much lower than that in BS cluster, i.e.
$q_{pm}=1-(1-q_{p0})^{1/6}$. In Baseline 1, 2 and 3, the control
policy are designed for imperfect CSIT.  The overall cell radius of
the system is 5000m\footnote{\textcolor{black}{5000 m is one of the
typical cell radius for LTE and LTE-advanced systems (e.g. rural
area)\cite{3GPP:25913}.}} in which Cluster 0 has radium of 2000m and
RS 1-6 are evenly distributed on a circle with radius 3000m as
illustrated in Figure \ref{Fig:model}. MSs randomly distribute in
the cell with $K_0=10$ MSs in Cluster 0 and $K_m=5$ in Cluster $m
(m=1,\cdots, 6)$. \textcolor{black}{The path loss model of BS-MS and
RS-MS is $128.1+37.6\log_{10}(R)$ dB, and path loss model of BS-RS
is $128.1+28.8\log_{10}(R)$ dB ($R$ in km). The lognormal shadowing
standard deviation is 8 dB. There are 64 subcarriers with 4
independent subchannels. The small scale fading follows
$\mathcal{CN}(0,1)$. We set up our simulation scenarios according to
the practical settings \cite{Std:4G-EMD:16m}.} The average
interference constraint to the PU is 0 dB.  Each point in the
figures is obtained by averaging over $2000$ independent fading
realizations.

\textbf{System Performance versus PU Activities:} Figure
\ref{Fig:qp-new} illustrates the PFS objective
$\sum_{k}\log{R_{k}}$\footnote{\textcolor{black}{$\sum_{k}\log{R_{k}}$
is the PFS optimization objective which is a good indication on the
tradeoff between throughput and fairness.}} (average sum-log-rate
{\em successfully received} by each MS) and access
probability\footnote{\textcolor{black}{Access probability is the
probability that a MS on the cell edge is allowed to receive data on
at least one subchannel in a scheduling slot.}} of MSs in Cluster
$m$ ($m=1,\cdots,M$) versus PU activity $q_{p}$ at receive $SNR=10$
dB and $\sigma_e^2=0.01$. $\sum_{k}\log{R_{k}}$ and access
probability decrease with the increase of PU activities. It can be
observed that our proposed scheme provides much greater access
probability as well as fairness/throughput performance for the MSs
at the cell edge compared with baseline 2 and 3 over a wide range of
PU activities. This performance gain is contributed by the
conventional RS path-loss gain as well as the increase in the access
opportunity for MS at the edge.

Figure \ref{Fig:bar-new} illustrates histogram of the average
goodput of MSs  (average data rate successfully received by the MSs)
at various distance from the BS at receive $SNR=10$ dB and
$\sigma_e^2=0.01$. It can be observed that baseline 2 can deliver
large system goodput only for those MSs close to the BS. It has very
low access probability and average goodput for those far-away
mobiles, causing severe fairness issues. However, there is a
significant gains in the system goodput of far-away MSs in the
proposed system and baseline 1, illustrating both the throughput and
fairness advantage of the system with RSs. Furthermore, the proposed
scheme outperforms SSA scheme in baseline 1.
\textcolor{black}{Figure \ref{Fig:cdf} illustrates the corresponding
CDF of average goodput of MSs at various distance. The low average
goodput regime ($x$-axis) demonstrates the performance of cell-edge
users: the larger probability ($y$-axis) in low average goodput
regime the larger average goodput of the cell-edge users.  It can be
observed that our proposed scheme brings better performance (larger
average goodput to the cell-edge users and sum average goodput of
all users ) compared with the baselines.}

\textbf{System Performance versus Receive SNR:} Figure
\ref{Fig:pwr-new} illustrates the PFS objective
$\sum_{k}\log{R_{k}}$ (average sum-log-rate {\em successfully received} by each MS) and access probability of MSs in Cluster m
($m=1,\cdots,M$) versus receive $SNR$. It can be observed that our
proposed design has significant gain over the baseline 1, 2 and 3
systems. The gain is more prominent at low SNR region because the
conventional RS reduce the path loss greatly and utilizes the
limited power more efficiently.

\textbf{System Performance versus the Number of MSs:} Figure
\ref{Fig:largek-new} illustrates $\sum_{k}\log{R_{k}}$ (average sum-log-rate {\em successfully received} by each MS) versus the
number of MSs in a cell at receive $SNR=10$ dB and
$\sigma_e^2=0.01$. The ratio between $K_0$ and $K_m$ is kept
constant. While the proposed scheme has the best performance over
baseline 2 and 3, the performance of all the three schemes increases
with K, which demonstrated the multi-user diversity gain in the
system.

\textbf{System Performance versus CSIT quality:} Figure
\ref{Fig:csiterror-new} illustrates the average system goodput (average data rate {\em successfully received} by the MS)
versus CSIT quality. The performance gain of the proposed scheme
versus baseline 1 illustrates the robustness of the proposed scheme
w.r.t. CSIT errors. On the other hand, comparison between baseline 2
and baseline 0 illustrated that it is very important to take CSIT
errors into the design. Baseline 0 has very poor performance because
there are a lot of error packets due to channel outage.

\section{Conclusion}\label{Section:Conclusion}
In this paper, we have proposed the design of downlink two-hop relay-assisted cognitive OFDMA system, which has the
cluster-based architecture and dynamically shares the spectrum of PU systems. Optimal decentralized algorithms have
been derived for joint rate and power control, and subchannel allocation at the RS and the BS respectively. These
algorithms maximize the weighted system goodput where proportional fair is included as a special case. The solution
processed local system state measurement at the BS and the RS to compute (locally) the power, rate and subchannel
allocations of the BS and RS. Imperfect system state measurement has been taking into consideration to maintain robust
performance of the SU and the PU systems.  Significant throughput gains have been observed from simulation results. We
have also derived a simple (asymptotically optimal) control algorithm as well as the closed-form performance for PFS
for sufficiently large number of users.

\section*{Appendix A: Solution of Subproblem \ref{sp:rs} and Subproblem \ref{sp:bs}}
The gradient of $L_m$ in subproblem 1 vanishes at the maximum, so we have
\begin{eqnarray}
\nonumber \frac{\partial{L_m}}{\partial{p_{m,n,k}}}=0 \Rightarrow
p_{m,n,k}&=&\alpha_{m,n,k}\bigg(\frac{\frac{1}{4}\big((1-\epsilon)\beta_{m,n}w_{m,k}-\mu_k\big)}{ln2\big(\nu+\eta_{n}\tau^2_{m,n}(1-\beta_{m,n})\big)}-\frac{1}{\phi_{m,n,k}}\bigg)^{+}\\
\nonumber \frac{\partial{L_m}}{\partial{\alpha_{m,n,k}}}=0
\Rightarrow
X_{m,n,k}&\triangleq&\frac{1}{4}\big((1-\epsilon)\beta_{m,n}w_{m,k}-\mu_k\big)\bigg(log_2(1+\frac{p_{m,n,k}\phi_{m,n,k}}{\alpha_{m,n,k}})\nonumber\\
& &
-\frac{p_{m,n,k}\phi_{m,n,k}}{ln2\(\alpha_{m,n,k}+p_{m,n,k}\phi_{m,n,k}\)}\bigg)=\lambda_{n}
\end{eqnarray}
and  $X_{m,n,k}$ can be interpreted as marginal benefit of extra
bandwidth. For a particular $\nu$, if there is a unique
$k^{*}=\arg \max\{X_{m,n,k}\}$ for some $n$, time-sharing will not
happen in this subchannel.
\begin{eqnarray}
\nonumber \alpha_{m,n,k} &=& \left\{ \begin{array}{ll} 1, &
X_{m,n,k}=\max_{k} \big\{X_{m,n,k}\big\} >0\\
0, & \textrm{otherwise}
\end{array} \right.
\end{eqnarray}
Since for each given $\mu$, $X_{m,n,k}$ is a function of the CSI
$\phi_{m,n,k}$, they are independent random variable. As a result,
there is probability 1 that one subchannel is assigned to a single
user.

We use the subgradient method to update the multipliers as follows
\begin{eqnarray}
\nonumber\lambda_n(i+1)&=&\bigg[\lambda_n(i)-\delta(i)\big(1-\sum_{k=1}^{K_m}
\alpha_{m,n,k}\big)\bigg]_\chi \forall n\\
\nonumber \nu(i+1)&=&\bigg[\nu(i)-\delta(i)\big(P_m-\sum_{n=1}^N
\sum_{k=1}^{K_m} p_{m,n,k}\big)\bigg]_\chi\\
\nonumber
\eta_n(i+1)&=&\bigg[\eta_n(i)-\delta(i)\big(I_{m,n}-\sum_{k=1}^{K_m}
(1-\beta_{m,n}) \tau^2_{m,n} p_{m,n,k} \big)\bigg]_\chi \forall n\\
\nonumber
\mu_k(i+1)&=&\bigg[\mu_k(i)-\delta(i)\big(R_{m,k}-\sum_{n=1}^{N}\frac{\alpha_{m,n,k}}{4}log(1+\frac{p_{m,n,k}\phi_{m,n,k}}{\alpha_{m,n,k}})\big)\bigg]_\chi
\forall k  \label{eq:update}
\end{eqnarray}
where $\{\delta(i)\}$ is a sequence of scalar step size and $\chi$
denotes the projection onto the feasible set, which contains all
non-negative real numbers. The iterative algorithm terminates when
the difference of two consecutive multipliers is less than a
terminating threshold. The subgradient update is guaranteed to
converge to the optimal multipliers $\lambda_n^{*}, \nu^{*},
\eta_n^*, \mu_k^*$.

We form the Lagrangian of Subproblem 2 as follows
\begin{eqnarray}
L_0&=&\widetilde{G}_0+\sum_{m=1}^M  \sum_{n=1}^N
(1-\epsilon)\beta_{0,n}\widetilde{G}_{n,m}^{*}-\sum_{n=1}^N\lambda_n\left(\sum_{k=1}^{K_0}
\alpha_{0,n,k} -1\right)\nonumber\\
&&-\nu\left(\sum_{n=1}^N \sum_{k=1}^{K_0} p_{0,n,k} -
P_0\right)-\sum_{n=1}^N \eta_n\left(\sum_{k=1}^{K_0}
(1-\beta_{0,n}) \tau^2_{0,n} p_{0,n,k} - I_{0,n}\right)
\end{eqnarray}
\begin{eqnarray}
\nonumber \frac{\partial{L_0}}{\partial{p_{0,n,k}}}=0
&\Rightarrow&
p_{0,n,k}=\alpha_{0,n,k}(\frac{\frac{1}{2}(1-\epsilon)\beta_{0,n}w_{0,k}}{ln2\big(\nu+\eta_{n}\tau^2_{0,n}(1-\beta_{0,n})\big)}-\frac{1}{\phi_{0,n,k}})^{+}\\
\nonumber \frac{\partial{L_0}}{\partial{\alpha_{0,n,k}}}=0
&\Rightarrow&
X_{0,n,k}\triangleq\frac{1}{2}(1-\epsilon)\beta_{0,n}w_{0,k}\bigg(log_2(1+\frac{p_{0,n,k}\phi_{0,n,k}}{\alpha_{0,n,k}})-\frac{p_{0,n,k}\phi_{0,n,k}}{ln2\big(\alpha_{0,n,k}+p_{0,n,k}\phi_{0,n,k}\big)}\bigg)=\lambda_{n}
\end{eqnarray}
\begin{eqnarray}
\nonumber \alpha_{0,n,k} &=& \left\{ \begin{array}{ll} 1, &
X_{0,n,k}=\max_{k} \big\{X_{0,n,k}\big\} >0\\
0, & \textrm{otherwise}
\end{array} \right.
\end{eqnarray}
where $w_{0,k} (k=1,\cdots,M)$ is the derivative on $\widetilde{G}_{m}^{*}$ w.r.t. $r_{0,n,m}$ which can be interpreted
as the equivalent weight of the $m$-th RS. As a result, we can use similar subgradient update procedure as in
(\ref{eq:update}) to obtain the multipliers $\lambda_n(i), \nu(i), \eta_n(i)$. Furthermore, when the data rates for the
relay stations are determined, the packet partition factors $\{d_{m,n,k}\}$ can be determined according to the
structure of $\tilde{G}_m^{**}(\cdot)$. Thus, select the best packet partition factors which can achieve the curve of
$\tilde{G}_m^{**}(\cdot)$.

\section*{Appendix B: Proof of Lemma \ref{lem:obj}}

According to Appendix A, it with probability $1$ that one subchannel is allocated to only one user or relay station in
phase one. Moreover, since channel between the base station and relay station is good enough, one subchannel is sufficient
to carry the data for the phase two transmission. Therefore, it with probability $1$ that one relay is allocated at
most one subchannel. Thus, for any relay station, there is only one positive value in the set of rate allocation
$\{r_{0,n,m}|\forall n\}$. Notice that $\tilde{G}_m^{**}(0)=0,$ we have
\begin{eqnarray}
\tilde{G}_m^{**}(\sum_n r_{0,n,m}t_{n,m}|\hat{\mathbf{S}}_m,\hat{\mathbf{H}}_m) &=& \sum_{n=1}^N
\tilde{G}_m^{**}(r_{0,n,m}t_{n,m}|\hat{\mathbf{S}}_m,\hat{\mathbf{H}}_m) \nonumber\\
& =& \sum_{n=1}^N t_{n,m} \tilde{G}_m^{**}(r_{0,n,m}|\hat{\mathbf{S}}_m,\hat{\mathbf{H}}_m) \quad \forall m. \nonumber
\end{eqnarray}
Hence,
\begin{eqnarray}
\mathbf{E}_{\mathbf{T}_m} \tilde{G}_m^{**}(\sum_n r_{0,n,m}t_{n,m}|\hat{\mathbf{S}}_m,\hat{\mathbf{H}}_m) &=&
\sum_{n=1}^N
 \mathbf{E}_{t_{n,m}} t_{n,m} \tilde{G}_m^{**}(r_{0,n,m}|\hat{\mathbf{S}}_m,\hat{\mathbf{H}}_m) \nonumber \\
&=& \sum_{n=1}^N
 (1-\epsilon) \beta_{0,n} \tilde{G}_m^{**}(r_{0,n,m}|\hat{\mathbf{S}}_m,\hat{\mathbf{H}}_m)
\end{eqnarray}
This complete the proof.

\section*{Appendix C: Proof of Lemma \ref{lem:linear}}

Without loss of generality, we consider the $m$-th cluster. Suppose there are $L$ QoS classes and denote $w_l$ as the
weight of the $l$-th QoS class. Since there is sufficiently large number of users in each cluster, the receiving SNR of
the selected users will be sufficiently large, therefore, equal power allocation is asymptotically optimal. Moreover,
since the relay station is only likely to pick up the best users (with the largest) from each QoS class, and the
channel fading of the best user tends to be a constant (e.g. $\ln K$) when the number of users is sufficiently large,
which subchannel is allocated to which QoS class become independent of the channel fading. Hence, the optimal resource
allocation is to do time-sharing among the $L$ class.

Let $g_{m,l}$ denote the maximum average weighted throughput of the $m$-th cluster if there are sufficient information
bits at the relay and only the users of the $l$-th QoS class are scheduled, and $\{r_{m,n,k}^l\}$ be the rate
allocation leading to the maximum average weighted throughput $g_{m,l}$. We define $r_{m,l}=\sum_{n,k}r_{m,n,k}^l$
denoting the corresponding total transmit data rate. These two parameters can be evaluated by each relay locally. We
first construct a function of $\mathcal{G}_m(r)$ ($\mathcal{G}_m:\mathcal{R}\rightarrow \mathcal{R}$) below (An example
of $\mathcal{G}_m(r)$ is shown in Figure \ref{Fig:linear}):
\begin{itemize}
\item Plot points $\{(r_{m,l},g_{m,l})|\forall l\}$ on a plane.
This refers to the points B and C in Figure \ref{Fig:linear}.

\item Let $\mathcal{H}$ be the convex hull of the points
$\{(r_{m,l},g_{m,l})|\forall l\}$ and $(0,0)$. This refers to the
triangle ABC in Figure \ref{Fig:linear}.

\item Define a region $\widetilde{\mathcal{H}}$ as
$\widetilde{\mathcal{H}} = \{(r,g)| \exists (r,g_e)\in
\mathcal{H}, \  g_e\leq g\}$. This refers to the area bounded by
line ABCD and x-axis in Figure \ref{Fig:linear}. Therefore, for
any given $r$, all the average weighted throughput in the set
$\{g|(r,g)\in \widetilde{\mathcal{H}}\}$ can is
achievable\footnote{An average weighted throughput $g$ is
achievable when there is a joint power, rate and subchannel
allocation at the cluster $m$ such that the average weighted
throughput is equal to $g$.} by the cluster $m$ using TDMA in each
frame.

\item $\mathcal{G}_m(r) = \max \{g|(r,g)\in
\widetilde{\mathcal{H}}\}$. This refers to the line ABCD in Figure
\ref{Fig:linear}.
\end{itemize}
An example of Therefore, $\widetilde{G}_{m}^{**}(r) \geq
\mathcal{G}_m(r)$. Moreover, since
$\frac{\widetilde{G}_{m}^{**}(r)-\mathcal{G}_m(r)}{\mathcal{G}_m(r)}\rightarrow
0$ for sufficiently large number of users in each QoS class, it's
asymptotically optimal to have $\widetilde{G}_{m}^{**}(r) =
\mathcal{G}_m(r)$.

\section*{Appendix D: Proof of Lemma \ref{lem:pfs}}

Without loss of generality, we consider the $m$-th cluster. Since the number of MSs in the $m$-th cluster $K_m$ is
sufficiently large, the sensing measure $\beta_{m,n}\rightarrow S_{m,n}$ and the system is working on the high SNR
regime. Therefore, the throughput gain of power allocation across the subchannels is negligible and we can simply assign
equal power to each available subchannel, thus $p_{m,n}= \frac{\beta_{m,n}}{\sum_{n=1}^N \beta_{m,n}}P_0$.

We first consider the case where $r \geq R_m = \sum_{m=1}^M\log_2\left(1+p_{m,n}l_{m,n,k}\varphi_{m,n,k}\right)$. In
this case, there are sufficient information bits at the relay for phase two transmission. Then the selected MS of the
$n$-th subchannel and the $m$-th cluster is given by $A_{m,n} = \arg\max_{k}
w_{m,k}\log_2\left(1+p_{m,n}l_{m,n,k}\varphi_{m,n,k}\right)$, and $ \widetilde{G}_{m}^{**}(r)= \sum_{n=1}^N
\frac{w_{m,A_{m,n}}\beta_{m,n}(1-\epsilon)}{4}\log_2(1+p_{m,n}l_{m,n,A_{m,n}}\phi_{m,n,A_{m,n}})$.

For the case where $r < R_m$, it's easy to see by linear
interpolation that $\widetilde{G}_{m}^{**}(r) \geq \\
\sum_{n=1}^N \frac{r w_{m,A_{m,n}}\beta_{m,n}(1-\epsilon)}{4R_m}\log_2(1+p_{m,n}l_{m,n,A_{m,n}}\phi_{m,n,A_{m,n}})$.
However, since the BS-RS link is sufficiently good, the BS always delivers $R_m$ bits to the $m$-th relay. Hence, we
can simply let $\widetilde{G}_{m}^{**}(r)=\\ \sum_{n=1}^N \frac{r
w_{m,A_{m,n}}\beta_{m,n}(1-\epsilon)}{4R_m}\log_2(1+p_{m,n}l_{m,n,A_{m,n}}\phi_{m,n,A_{m,n}})$, which does not affect
the scheduling results at the BS.

\section*{Appendix E: Proof of Theorem \ref{thm:aym}}

Due to page limitation, we provide a sketch of proof. When the RS-BS link is sufficiently good due to the existence of
line-of-sight path, the relay will always receive sufficiently information bits as long as there is one available
subchannel in cluster $0$, and the PFS algorithm in each relay cluster works as that in single cell systems with infinite backlog. Hence, we can follow the similar approach as in \cite{Caire:07} to prove that  when $K_c$ is
sufficiently large, the user selection is based on the small-scale channel fading, which leads to (\ref{eqn:asy-user}).

Since there are $N$ subchannels in the system, the probability the BS can not deliver packets to the relays is $q_p^N$.
Hence, in each cluster the probability one subchannel is used to deliver packet is $(1-q_p)(1-q_p^N)$. Again, by
following the similar approach as in \cite{Caire:07}, the $\overline{T}_{m,k}$ can be derived after some algebra.

\bibliographystyle{ieeetr}
\bibliography{MANET,SDMA,Huang,Ray,Lau_05-1017,cuiying-bib}

\begin{figure}
\centering
\includegraphics[height=7cm, width=15cm]{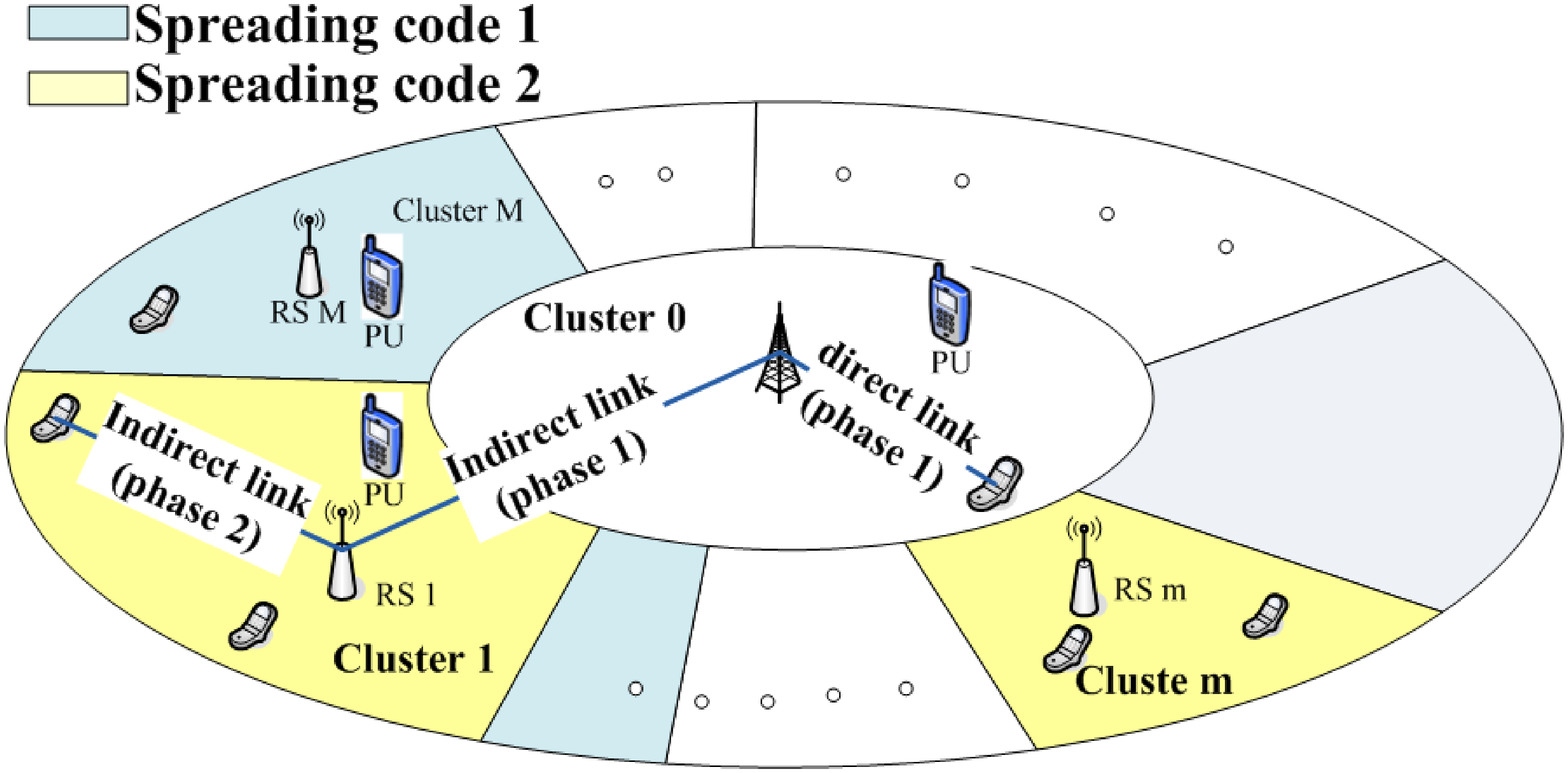}
\caption{Cluster based relay-assisted downlink OFDMA system.}\label{Fig:model}
\end{figure}

\begin{figure}
\centering
\includegraphics[height=5cm, width=15cm]{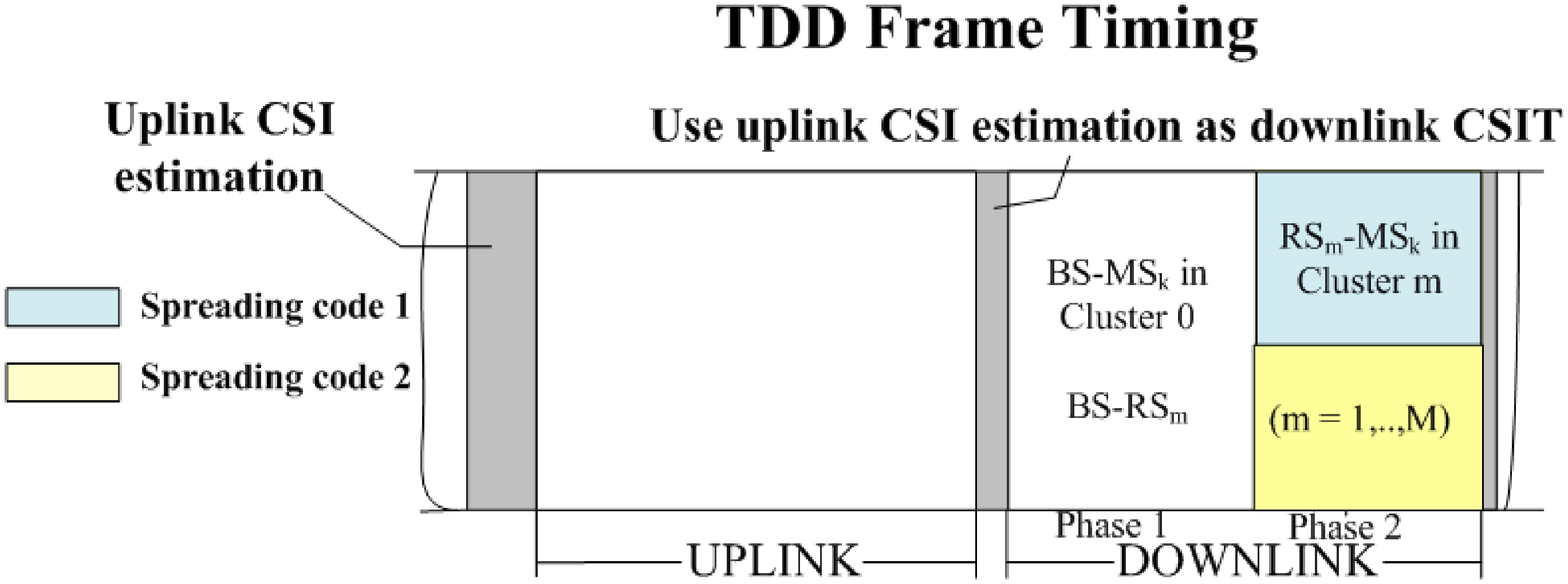}
\caption{A frame structure example for relay-assisted OFDMA system.}\label{Fig:frame}
\end{figure}

\begin{figure}
\centering
\includegraphics[height=10cm, width=12cm]{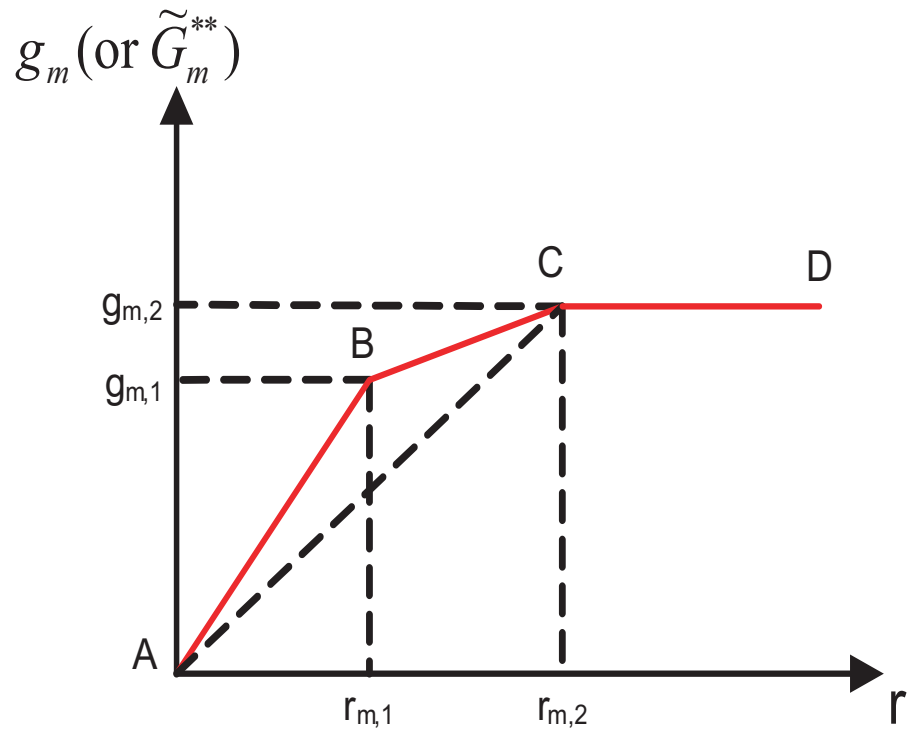}
\caption{An example of function $\widetilde{G}_{m}^{**}$ where
there are two QoS classes with weights $w_1$, $w_2$ and maximum
achievable data rate $r_{m,1}$ ,$r_{m,2}$. The x-axis is the
number of information bits the RS decoded in phase one, and the
y-axis is the maximum average weighted goodput achieved by this
RS. When the number of information bits is less than $r_{m,1}$,
only the first QoS class is scheduled; when it's larger than
$r_{m,1}$ but less than $r_{m,2}$, both two classes are scheduled
by TDMA; and when it's larger than $r_{m,2}$, only the second QoS
class is scheduled.}\label{Fig:linear}
\end{figure}

\begin{figure}
\centering
\includegraphics[height=8cm, width=10cm]{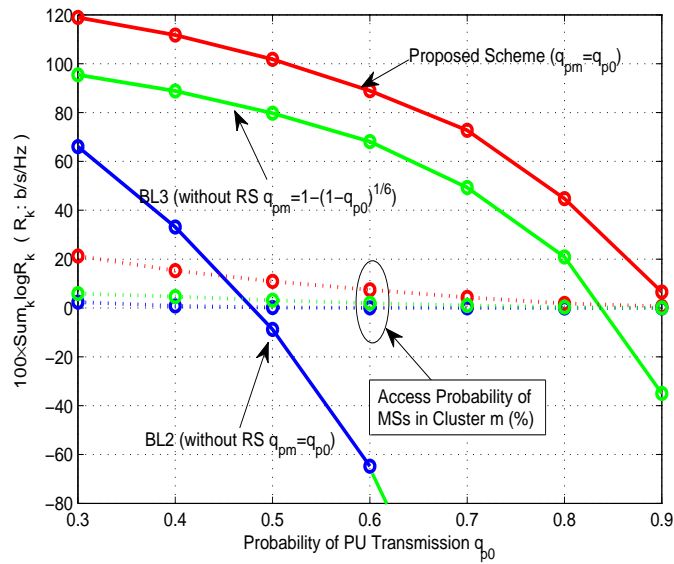}
\caption{\textcolor{black}{$\sum_{k}\log{R_{k}}$ and access
probability of MSs in Cluster m ($m=1,\cdots,M$) versus probability
of PU transmission $q_{p0}$.} $q_f=0.2$, $q_d=0.8$, $M$=6, $N$=4,
$K_{0}$=10, $K_{m}$=5, $I$=0 dB, receive $SNR=10$ dB,
$\sigma_e^2=0.01$. } \label{Fig:qp-new}
\end{figure}

\begin{figure}
\centering
\includegraphics[height=8cm, width=10cm]{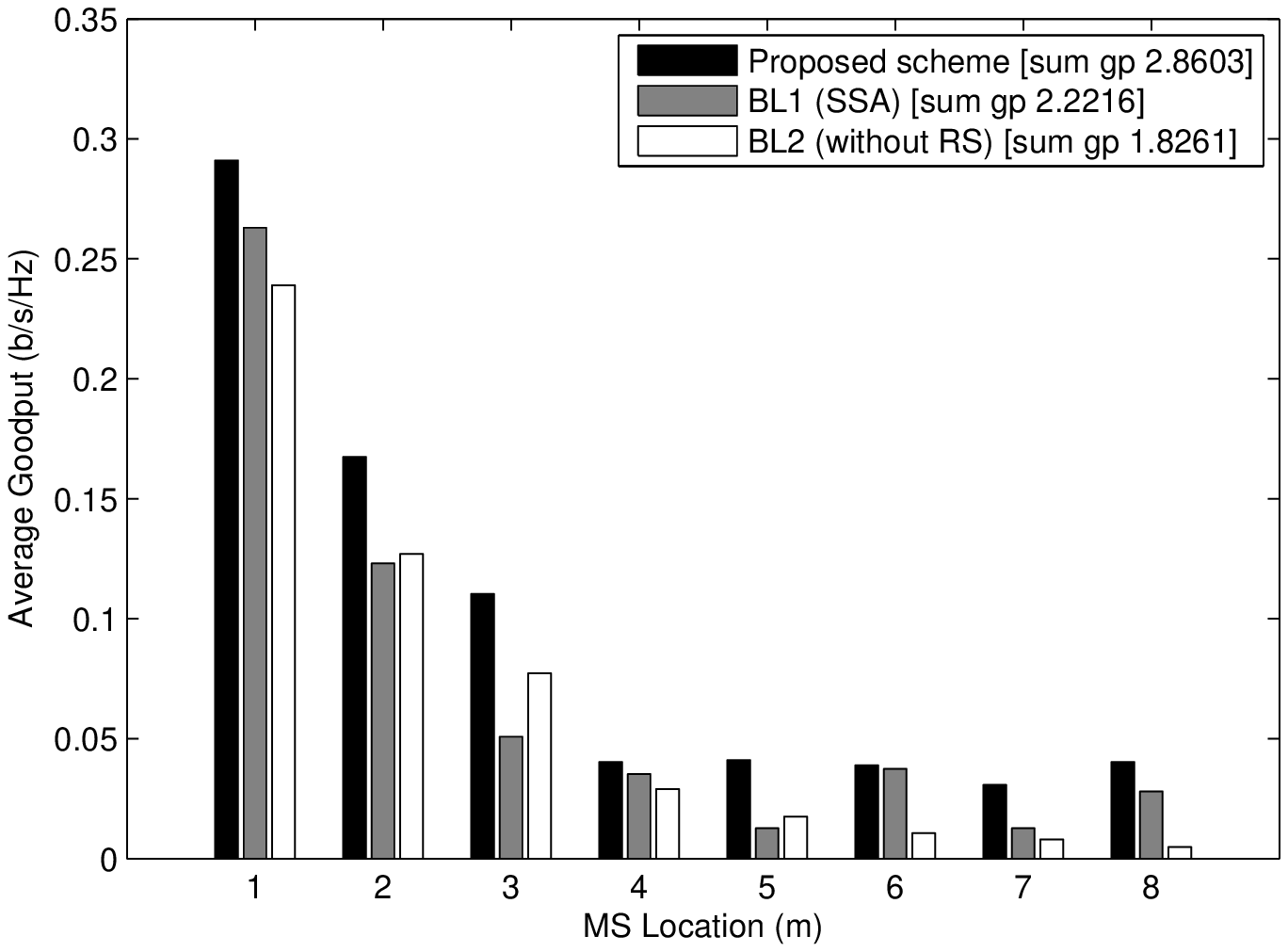}
\caption{\textcolor{black}{Histogram of the average goodput of MSs
(average data rate successfully received at the MSs) at various
distance from the BS.} $q_{p0}=p_{pm}=0.3$, $q_f=0.2$, $q_d=0.8$,
$M$=6, $N$=4, $K_{0}$=10, $K_{m}$=5, $I$=0 dB, receive $SNR=10$ dB,
$\sigma_e^2=0.01$. }\label{Fig:bar-new}
\end{figure}

\begin{figure}
\centering
\includegraphics[height=8cm, width=10cm]{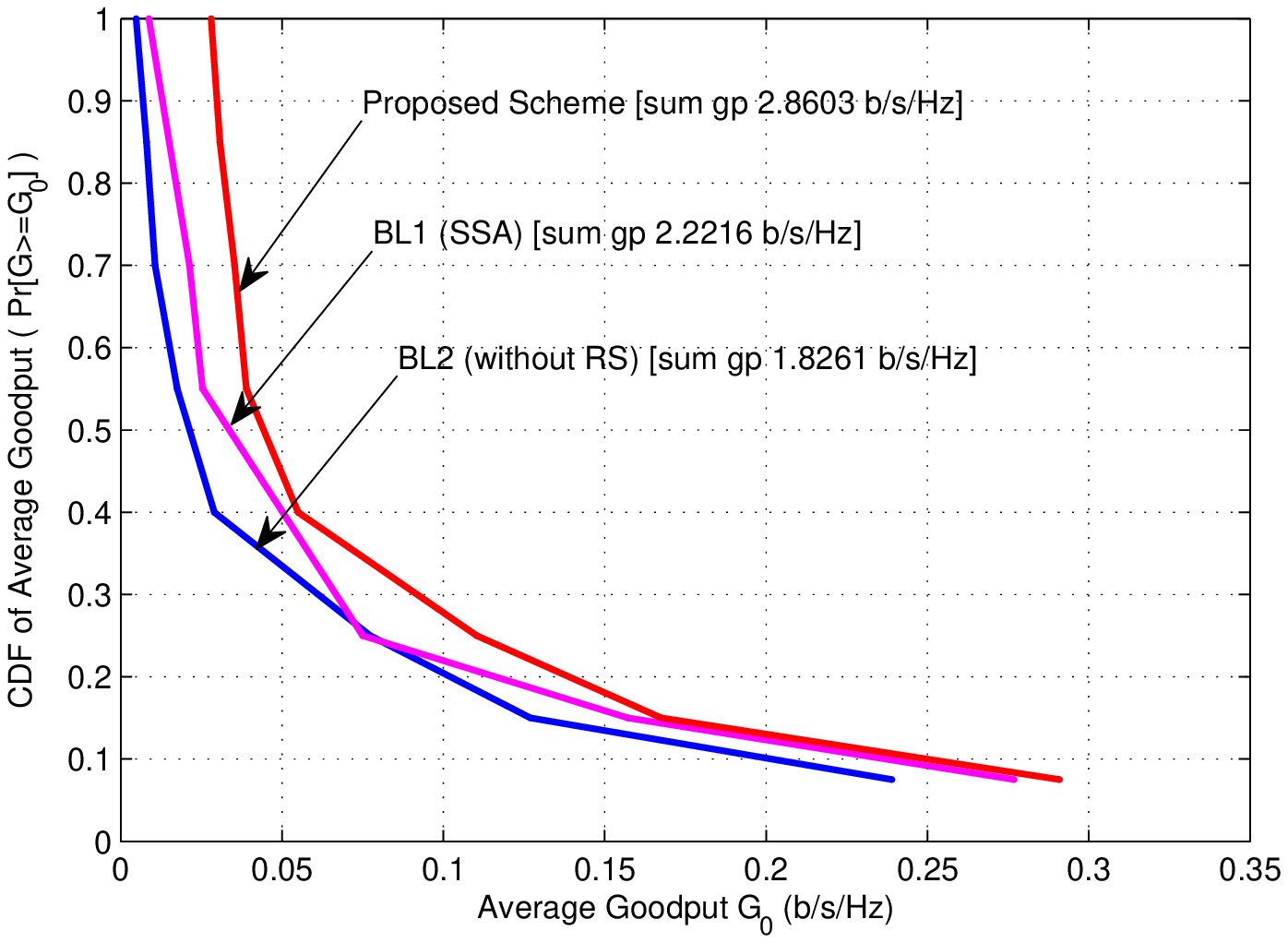}
\caption{\textcolor{black}{CDF of the average goodput of MSs
(average data rate successfully received at the MSs) at various
distance from the BS.}  $q_{p0}=p_{pm}=0.3$, $q_f=0.2$, $q_d=0.8$,
$M$=6, $N$=4, $K_{0}$=10, $K_{m}$=5, $I$=0 dB, receive $SNR=10$ dB,
$\sigma_e^2=0.01$. }\label{Fig:cdf}
\end{figure}

\begin{figure}
\centering
\includegraphics[height=8cm, width=10cm]{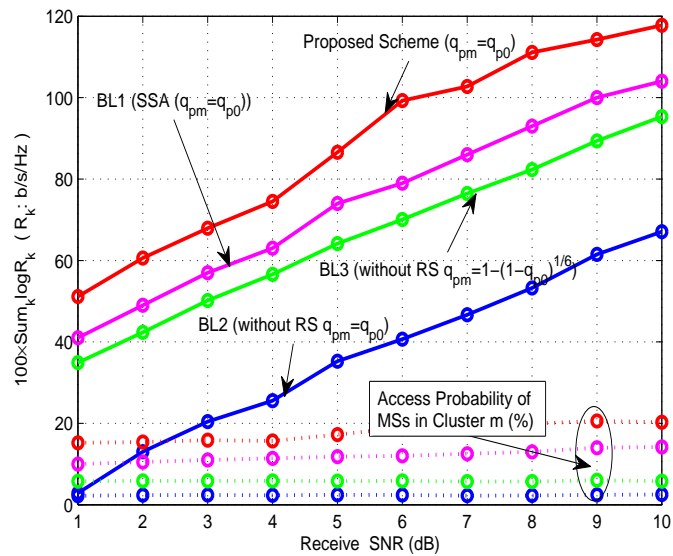}
\caption{\textcolor{black}{$\sum_{k}\log{R_{k}}$ and access
probability of MSs in Cluster m ($m=1,\cdots,M$) versus receive
$SNR$.} $q_{p0}=0.3$,$q_f=0.2$, $q_d=0.8$, $M$=6, $N$=4, $K_{0}$=10,
$K_{m}$=5, $I$=0 dB, $\sigma_e^2=0.01$. } \label{Fig:pwr-new}
\end{figure}

\begin{figure}
\centering
\includegraphics[height=8cm, width=10cm]{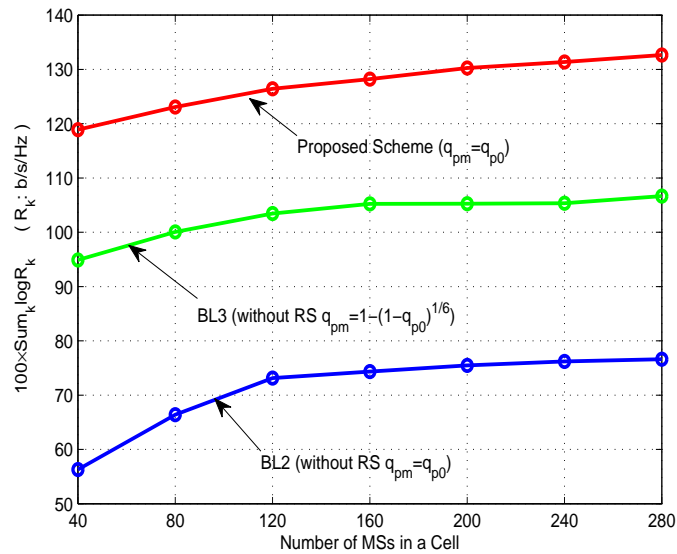}
\caption{\textcolor{black}{$\sum_{k}\log{R_{k}}$ versus the number
of MSs in a cell.} $q_{p0}=0.3$, $q_f=0.2$, $q_d=0.8$, $M$=6, $N$=4,
$I$=0 dB, receive $SNR=10$ dB, $\sigma_e^2=0.01$.}
\label{Fig:largek-new}
\end{figure}

\begin{figure}
\centering
\includegraphics[height=8cm, width=10cm]{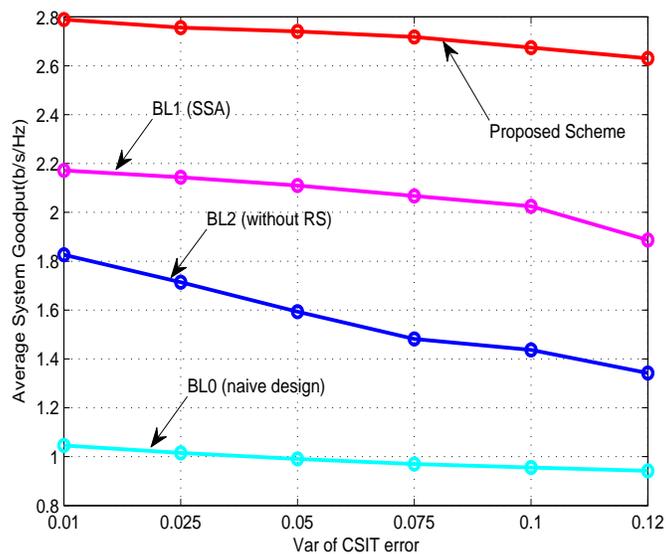}
\caption{\textcolor{black}{Average system goodput  (average data
rate successfully received at the MSs) versus CSIT quality.}
$q_{p0}=p_{pm}=0.3$, $q_f=0.2$, $q_d=0.8$, $M$=6, $N$=4, $K_{0}$=10,
$K_{m}$=5, $I$=0 dB, receive $SNR=10$ dB.} \label{Fig:csiterror-new}
\end{figure}

\end{document}